\documentclass[11pt,letterpaper,USenglish,cleveref,autoref,thm-restate]{article}

\usepackage{amsmath}
\usepackage{graphicx}
\usepackage{epstopdf}

\usepackage{amsthm}
\usepackage{amssymb} 
\usepackage{hyperref}
\usepackage{cleveref,lineno}
\usepackage{mathtools,amsmath} 
\usepackage{graphicx}
\usepackage{mathrsfs}
\usepackage{thmtools} 
\usepackage{soul,esvect}
\usepackage{epstopdf}
\usepackage{thm-restate,todonotes}
\usepackage{fullpage}
\usepackage{xspace}
\usepackage{enumitem}
\usepackage{authblk}

\usepackage{fix-cm}
\usepackage{multirow}
\usepackage[table,xcdraw]{xcolor}

\usepackage{algorithm}
\usepackage[noend]{algpseudocode}

\usepackage{lineno}

\usepackage[]{todonotes}

\newtheorem{lemma}{Lemma}
\newtheorem{observation}{Observation}
\newtheorem{definition}{Definition}

\newtheorem{claim}{Claim}

\newcommand\abs[1]{\lvert #1\rvert}
\newcommand{\tw}{\mathsf{tw}}
\newcommand{\pw}{\mathsf{pw}}
\newcommand{\ppw}{\mathsf{ppw}}

\newcommand{\ipw}[1]{{\mathsf{w}(\mathcal P[#1])}}

\newcommand{\ld}{\textsf{Directed-LDD}}
\newcommand{\ldd}{\textsf{Directed-LDD }}

\newcommand{\rp}{\textsf{Scissors-Carving }}
\newcommand{\rpp}{\textsf{Scissors-Carving}}

\newcommand{\dem}{\textsf{dem}}

\newcommand{\outb}[2]{B_{#1}^+(#2)}
\newcommand{\inb}[2]{B_{#1}^-(#2)}

\algnewcommand{\commentline}[1]{%
  \Statex /\textasteriskcentered\text{ } \textbf{#1} \textasteriskcentered /%
}

\title{Directed Low Diameter Decomposition for Structured Digraphs}
\author{Shinwoo An\thanks{
Email: \texttt{shinwoo.an@biu.ac.il}. This work was supported by the Isreal Science Foundation grant $\#$1042/22, Israel Science Foundation grant $\#$800/22, and US-Israel Binational Science Foundation Grant $\#$2022/131.} }
\author{Arnold Filtser\thanks{Email: \texttt{arnold.filtser@biu.ac.il}. This work was supported by the Isreal Science Foundation grant $\#$1042/22.}}
\affil{Department of Computer Science, Bar-Ilan University, Israel.}

\begin{document}
\date{}
\maketitle

\thispagestyle{empty}
\begin{abstract}
    Low diameter decompositions, or LDDs for short, are a fundamental primitive in the design of efficient graph algorithms.
    Roughly speaking, an LDD is a distribution over partitions of the vertices into bounded-diameter clusters such that nearby vertices are likely to be clustered together. 
    Recently, there has been growing interest in lifting the notion of LDDs into \emph{directed graphs}. 
    In particular, there are two natural directed analogues. 
    The first is a directed LDD, where after removing a random subset of edges, every strongly connected component has a small diameter. 
    The second is a quasipartition, which imposes the stronger requirement that whenever one vertex can still reach another after the edge removal, the two vertices must be close in the original directed metric. 
    Every quasipartition yields an LDD, but the converse does not necessarily hold. 

    In this work, we initiate the systematic study of LDDs in structured directed graphs. As our first main result, we show that any directed graph with pathwidth $\pw$ admits an $(O(\pw), \Delta)$-LDD. This improves upon the previous best-known $(2^{O(\pw^2)}, \Delta)$-LDD construction, which was implicitly derived from the quasipartition result of Salmasi, Sidiropoulos, and Sridhar [SODA'19].

    As our second result, we show that the integrality gap of the Directed Non-Bipartite Sparsest-Cut LP relaxation on an $n$-vertex graph with treewidth $\tw$ is $O(\tw \log n)$. This improves upon the $O(\tw\log^2 n)$ bound of M\'emoli, Sidiropoulos, and Sridhar [ICALP'16, Algorithmica'18]. 
    We obtain this result through the refined analysis of the quasipartition construction of M\'emoli et al. for bounded treewidth graphs.
\end{abstract}

\newpage
\setcounter{page}{1}
\section{Introduction}
Decomposing graphs into well-structured pieces is a central paradigm in the design of graph algorithms.
In particular, low diameter decompositions (LDDs), which are a distribution over partitions of the vertices into bounded-diameter clusters such that nearby vertices are likely to be clustered together, have played a central role in undirected graphs. They serve as key primitive underlying a wide range of problems, a partial list of examples include: metric embedding~\cite{Bar96,FRT04,krauthgamer2005measured,Rao99,Filtser25}, multi-commodity flow~\cite{KPR93,LR99}, Steiner point removal~\cite{EGKRTT14}, near-linear SDD solvers~\cite{BGKMPT14}, Lipschitz extension~\cite{LN05}, spanners~\cite{filtser2022light,har2023reliable}, and more.
In undirected graphs, low diameter decompositions have been extensively studied not only in general graphs, but also in structured graph classes~\cite{KPR93,FT03,AGMW10,AGGNT19,Fil19padded,filtser2025optimal,CF25}, where additional structure often leads to improved guarantees.

Recently, there has been a growing interest in extending low diameter decompositions to 
\emph{directed graphs}~\cite{BGW20,BNW25,BFHL25,Li25}, despite the inherent technical challenges compared to the undirected setting.
A notable application of this line of work is a recent near-linear time algorithm for the negative-weight single-source shortest path (SSSP) problem~\cite{BNW25}.
A \emph{strongly connected component} of a directed graph $G$, an SCC in short, is a maximal subgraph $H$ of $G$ such that for every pair $u,v$ of vertices in $H$, both $u$ is reachable from $v$ and $v$ is reachable from $u$. 
For vertices $u,v$ of $G$, let $d_G(u,v)$ denote the length of a shortest directed path from $u$ to $v$ in $G$, with $d_G(u,v)=\infty$ if no such path exists. For a subgraph $H\subseteq G$, the strong diameter of $H$ is $\max_{u,v\in V(H)} d_H(u,v)$, whereas the weak diameter of $H$ is $\max_{u,v\in V(H)} d_G(u,v)$. Here, $V(H)$ denotes the vertex set of $H$.

\begin{definition}[Directed Low-Diameter Decomposition,~\cite{BGW20,BNW25}]\label{def:ldd}
Given an edge-weighted directed graph $G=(V,E,w)$ and a parameter $\Delta>0$, a $(\beta, \Delta)$-LDD is a distribution over subsets
$S\subset E$ of edges, called a cutset, such that
\begin{itemize} \setlength{\itemsep}{0.1pt}
    \item[1.] every strongly connected component in $G\setminus S$ has a weak diameter of at most $\Delta$, and
    \item[2.] for every edge $e\in E$, $\Pr[e\in S]\leq \beta\cdot \frac{w(e)}{\Delta}$.
\end{itemize} 
We call $\beta$ the loss parameter.
\end{definition}

On the other hand, a parallel line of work has focused on \emph{quasipartitions} for directed graphs~\cite{memoli2018quasimetric,kawarabayashi2022embeddings,SSS19}, motivated by their connection to the integrality gap of the directed sparsest cut problem and the directed multicut problem~\cite{charikar2006directed}.
\begin{definition}[(Random) Quasipartition~\cite{memoli2018quasimetric}] \label{def:quasi} 
    Given an edge-weighted directed graph $G=(V,E,w)$ and a parameter $\Delta>0$, a $(\beta, \Delta)$-quasipartition is a distribution over a cutset $S\subseteq E$ such that
    \begin{itemize}  \setlength{\itemsep}{0.1pt}
        \item for every $u$-$v$ path in $G\setminus S$, $d_G(u,v) \leq \Delta$, and
        \item for every edge $e\in E$, $\Pr[e\in S]\leq \beta\cdot \frac{w(e)}{\Delta}$.
    \end{itemize}
    We call $\beta$ the loss parameter.
\end{definition}

Notably, quasipartitions impose strictly stronger constraints than LDDs, and thus any $(\beta, \Delta)$-quasipartition immediately yields a $(\beta,\Delta)$-LDD, but not vice versa. 
Intuitively, in the quasipartition construction, one has to enforce that the distance between endpoints of every valid path is at most $\Delta$. In contrast, the LDD construction only restricts the distance between endpoints of the paths within a 
\emph{strongly connected component}. 
One specific example is directed acyclic graphs (DAGs), which are directed graphs avoiding directed cycles. It is easy to observe that every DAG itself is a $(0, \Delta)$-LDD for every parameter $\Delta>0$ since its SCCs are isolated singletons. On the other hand, constructing a quasipartition with a small loss parameter on a DAG remains highly non-trivial, as a DAG still might contain arbitrarily long paths that must be cut.
This separation is also reflected in a polynomial gap over the achievable loss parameters. While every directed graph on $n$ vertices admits an $(O(\log n \cdot \log\log n), \Delta)$-LDD~\cite{BFHL25},
there exists an $n$-vertex directed graph that does not admit any $(o(n^{1/7}\log^{4/7} n), \Delta)$-quasipartition~\cite{memoli2018quasimetric}.

Quasipartitions and LDDs are not merely theoretical concepts. Rather, they play a crucial role in driving  distance-preserving graph sparsification scheme. Here, the goal is to obtain a simplified representation of the given graph while approximately preserving its distances. 
A number of such structures for directed graphs have been studied: including spanners~\cite{berman2011improved}, distance preservers~\cite{bodwin2021new}, hopsets~\cite{KP22}, and more. From the aforementioned lower bound for the quasipartition construction,~\cite{memoli2018quasimetric} showed that one cannot efficiently embed a directed graph into \emph{quasiultrametric space}, which is equivalent to the shortest path metric in directed trees. 
In contrast, by changing the viewpoint to LDDs, Assadi, Hoppenworth, and Wein~\cite{AHW25} showed that one can (stochastically) embed any directed graph into two DAGs with distortion $O(\log^3 n)$. Very recently, the distortion bound has been improved to $O(\log n\cdot \log\log n)$ by Filtser~\cite{Fil26}.

\medskip
In structured directed graph classes, the distinction between the two decompositions remains poorly explored. 
First, all known LDD results for structured directed graphs have an indirect nature:
Instead of being tailored directly to LDDs, they are either obtained via quasipartition constructions (bounded pathwidth and bounded treewidth~\cite{memoli2018quasimetric,SSS19}) or simply bypassed by applying general graph bounds~\cite{BFHL25} (planar graphs~\cite{kawarabayashi2022embeddings}).
In addition, there are no techniques that are specifically tailored to constructing LDDs in this setting.

Second, among the prior works that are indirectly derived from the quasipartition constructions, the resulting bounds fail to yield a clear improvement over the general graph bounds.
Specifically, the loss parameter either suffers from a polynomial dependence on treewidth~\cite{memoli2018quasimetric} or even worse, an exponential dependency on pathwidth~\cite{SSS19}. Consequently, these results surpass the general bounds of~\cite{BFHL25} only when the parameters are strictly sublogarithmic in the graph size.

To overcome these barriers, a natural starting point is to focus on
pathwidth~\cite{robertson1983graph}- one of the most standard and well-studied graph parameters which intuitively measures the degree to which a given graph resembles a path. Indeed, based on their structural advantage, bounded pathwidth graphs admit a wide range of efficient algorithms for fundamental graph problems~\cite{AFGN22,fomin2018fully,gutin2016mixed,kante2022linear,lampis2024structural}.
From the perspective of LDDs, such graphs also enjoy improved guarantees on the loss parameter compared to general graphs in the undirected setting~\cite{AGGNT19}.
Hence, it is natural to ask whether the same phenomenon holds in the directed setting.

Motivated by these insights and the gap between the undirected and directed settings, we initiate the study of LDDs in structured directed graphs, especially bounded pathwidth graphs, moving beyond those obtained indirectly via quasipartition approaches. 
Our main result is that every directed graph of pathwidth $\pw$ admits an $(O(\pw), \Delta)$-LDD. This significantly improves the previous best-known result of $(2^{O(\pw^2)}, \Delta)$-LDD~\cite{memoli2018quasimetric} in this setting.

\begin{restatable}{theorem}{pathwidthldd}\label{thm:ldd}
    Let $G$ be an edge-weighted directed graph and $\pw(G)=k$. Then, for any $\Delta>0$, there is a $(O(k), \Delta)$-LDD of $G$.
\end{restatable}

\begin{table}[h]
\centering
\renewcommand{\arraystretch}{1.3}
\begin{tabular}{|l|l|l|l|}
\hline
\textbf{Graph Family} & \textbf{Partition type} &\textbf{ Loss parameter} $\boldsymbol{\beta}$ & \textbf{Reference} \\ \hline
General graphs & \multirow{5}{*}{Undirected LDD} & $\Theta(\log n)$ (UB and LB)& \cite{Bar96} \\ \cline{1-1} \cline{3-4}

Genus $g$ &  & $\Theta(\log g)$ &\cite{lee2009geometry} \\ \cline{1-1} \cline{3-4}
Pathwidth $\pw$ &  & $\Theta(\log \pw)$&\cite{AGGNT19} \\ \cline{1-1} \cline{3-4}
Treewidth $\tw$ &  & $\Theta(\log \tw)$&\cite{filtser2025optimal} \\ \cline{1-1} \cline{3-4}
$K_r$-minor free &  & $\Theta(\log r)$& \cite{CF25} \\ \hline

General digraphs & Directed LDD & $O(\log n \cdot \log\log n)$&\cite{BFHL25} \\ \hline
General digraphs & Quasipartition &$\Omega(n^{1/7}\log^{4/7}n)$ (LB)&\cite{chuzhoy2009polynomial,memoli2018quasimetric} \\ \hline
Planar & Quasipartition& $O(\log^2 n)$ &\cite{kawarabayashi2022embeddings} \\ \hline
Treewidth $\leq 2$ &Quasipartition & $O(1)$& \cite{SSS19} \\ \hline
Treewidth $\tw$& Quasipartition& $O(\tw \log n)$&\cite{memoli2018quasimetric} \\ \hline
Pathwidth $\pw$& Quasipartition& $2^{O(\pw^2)}$& \cite{SSS19} \\ \hline

\rowcolor[HTML]{EFEFEF}
Pathwidth $\pw$ & Directed LDD& $O(\pw)$& \Cref{thm:ldd} \\ \hline
\end{tabular}
\caption{Previous and new results on graph partition in directed graphs}
\end{table}

Classical constructions of (undirected/directed) LDDs rely on the ball carving technique~\cite{Bar96}, where we recursively compute clusters by removing regions defined by distance thresholds around a center. 
In graphs of bounded pathwidth, prior work for undirected graphs~\cite{AGGNT19} bypasses the general bound of~\cite{Bar96} via a path-carving approach: repeatedly choosing a shortest path
as a backbone and removing a distance-threshold neighborhood around it, exploiting
the path-like structure of the graph. Provided that one can efficiently remove this neighborhood~\cite{AGGNT19,kamma2017metric}, this approach yields an improved loss bound compared to~\cite{Bar96}.

On the other hand, in general directed graphs, one natural difficulty is that, once we carve a nearby region around a center as a cluster, we cannot guarantee the diameter of the cluster. To overcome this difficulty, all prior works that directly compute LDDs~\cite{BGW20,BNW25,BFHL25,Li25} recursively apply the ball carving approach for each carved cluster. It seems difficult to eliminate the $\log n$ dependency on the loss parameter under this approach.

Our main technical contribution is a novel carving procedure with respect to a shortest path, which we call \emph{scissors carving}, that avoids this difficulty.
In essence, our approach carefully decomposes a shortest path into shorter subpaths, and applies the ball carving to each segment in a ``local" sense. 
The key insight is that, we can efficiently carve away a path whose length is sufficiently short\footnote{See Section~\ref{sec:overview-ours} for more details.}. 
In addition, we ensure that any carved cluster from such short paths has a small diameter.
We believe that our technique provides a first step toward developing LDD techniques tailored to structured directed graphs, and that the technique itself is of independent interest.

\medskip
As a secondary contribution, we revisit quasipartitions in directed graphs of bounded treewidth \cite{memoli2018quasimetric}, 
and improve the $O(\tw \log^2 n)$ bound on the integrality gap of the directed non-bipartite sparsest cut~\cite{memoli2018quasimetric} by a factor of $O(\log n)$, through a better analysis. 
Previous and new results about the integrality gap are summarized in Table~\ref{tab:quasi}.
See Section~\ref{sec:quasipartition} for the problem definitions and details of prior works on it. 

\begin{restatable}{theorem}{flowcutgap} \label{thm:flowcutgap}
    Let $G$ be an $n$-vertex directed graph of treewidth $\tw$. The integrality gap (also called the flow-cut gap) of the Directed Non-Bipartite Sparsest-Cut LP relaxation on $G$ is $O(\tw \log n)$.
\end{restatable}

\begin{table}[h] 
\centering
\renewcommand{\arraystretch}{1.3}
\begin{tabular}{|l|l|l|l|} 
\hline
\textbf{Graph Family} & \textbf{Integrality Gap} & \textbf{Ref.} & \textbf{Notes} \\ \hline
\multirow{3}{*}{General digraphs} & $O(\sqrt n)$ & \cite{hajiaghayi2006n}& UB \\ \cline{2-4} 
 & $\tilde O(n^{11/23})$ &\cite{agarwal2007improved} & UB \\ \cline{2-4}
 & $\tilde \Omega(n^{1/7})$ & \cite{chuzhoy2009polynomial} & LB \\ \hline
Planar  & $O(\log^3 n)$ & \cite{kawarabayashi2022embeddings}& UB \\ \hline
Pathwidth $\pw$&$2^{O(\pw^2)}\cdot \log n$ &\cite{memoli2018quasimetric}+\cite{SSS19} & UB \\ \hline
Treewidth $\tw$& $O(\tw \log^2 n)$  &\cite{memoli2018quasimetric} & UB\\ \hline
\rowcolor[HTML]{EFEFEF}
Treewidth $\tw$ & $O(\tw \log n)$  & \Cref{thm:flowcutgap} & UB\\ \hline
\end{tabular}
\caption{Previous and new results on the integrality gap of the directed non-bipartite sparsest cut LP. Here, $\tilde O(\cdot)$ hides polylog $n$ factor.}
\label{tab:quasi}
\end{table}

\section{Preliminaries} \label{sec:prelim}
For a positive integer $n$, we denote by $[n]$ the set of all positive integers at most~$n$.
\paragraph{Graph notations.}
Let $G = (V, E, w)$ be an edge-weighted directed graph, with $n$ vertices. We often use $V(G)$ and $E(G)$ to denote the vertex set and the edge set in $G$, respectively.
For a set $A$ of vertices in $G$, we use $G\setminus A$ to denote the graph obtained from $G$ by removing all vertices in $A$ and their adjacent edges.
We use $G[A]$ to denote a \emph{induced subgraph} of $G$ with respect to $A$. That is, $G[A]=G\setminus (V\setminus A)$.
Similarly, for a set $B$ of edges in $G$, we use $G\setminus B$ to denote the graph obtained from $G$ by removing all edges in $B$.
Two vertices $u,v$ in $G$ are \emph{strongly connected} if there exist both a $u$-$v$ path and a $v$-$u$ path in $G$.
A graph $G$ is \emph{strongly connected} if there is a $u$-$v$ path for every $u,v\in V$. 
A \emph{strongly connected component} of $G$, an SCC in short, is a maximal strongly connected subgraph of $G$. 
For a path $P$ in $G$, the \emph{length} of $P$ is defined as the total edge weight in $P$.
We use $d_G(u,v)$ to denote the \emph{distance} between $u$ and $v$ in $G$.
The \emph{diameter} of $G$ is defined as the maximum value of $d_G(u,v)$ among every pair $(u,v)\in V\times V$.

\paragraph{Treewidth and Pathwidth.}
A \emph{tree-decomposition} (resp. \emph{path-decomposition}) of an \emph{undirected} graph $G$ is a pair $(T, \{B_t\}_{t\in V(T)})$ consisting of a tree (resp. path) $T$ and a family of sets $\{B_t\}_{t\in V(T)}$ of vertices in $G$ such that 
\begin{enumerate}  \setlength{\itemsep}{0.1pt}
    \item[1.] $V(G)=\bigcup_{t\in V(T)}B_t$, 
    \item[2.] for every edge $uv \in E(G)$, there exists a node $t \in V(T)$ such that $\{u, v\} \subseteq B_t$, and 
    \item[3.] for every vertex $v \in V(G)$, the set $\{t\in V(T):v\in B_t\}$ induces a connected subtree (resp. subpath) of $T$. \label{TD:connectedness}
\end{enumerate}
The \emph{width} of a tree-decomposition (resp. path-decomposition) $(T, \{B_t\}_{t \in V(T)})$ is defined as $\max_{t\in V(T)}\abs{B_t}-1$.
The \emph{treewidth} $\tw(G)$ (resp. \emph{pathwidth} $\pw(G)$) of a graph $G$ is the minimum width among all tree decompositions (resp. path decompositions) of $G$.

For a \emph{directed} graph $G=(V,E)$, we consider the undirected graph $\bar G$ obtained by ignoring the edge directions and removing any parallel edges. Formally speaking, $\bar G$ has a vertex set $V(\bar G)=V$ and an edge set $E(\bar G)=\{\{u,v\}\mid (u,v)\in E \text{ or } (v,u)\in E\}$.
We say that $G$ has treewidth (resp. pathwidth) $k$ if its underlying undirected graph $\bar G$ has treewidth (resp. pathwidth) $k$.

\section{Technical Overview}
\paragraph{Low diameter decomposition (LDD) for undirected graphs.}
Let $G=(V,E,w)$ be an undirected graph. A $(\beta, \Delta)$-LDD of $G$ is a distribution over vertex partitions of $V$ such that in every partition, each cluster has a diameter of at most $\Delta$, and for every pair $u,v\in V$, the probability that $u,v$ belong to different clusters is at most $\beta\cdot \frac{d_G(u,v)}{\Delta}$.
We call $\beta$ the \emph{loss parameter}. Every $n$-vertex graph admits an $(O(\log n), \Delta)$-LDD for any $\Delta>0$, which is tight~\cite{Bar96}.

Ball carving is a classical technique for computing an LDD. For a vertex $u\in V$ and parameter $R>0$, a \emph{ball} $B_G(u,R)$ is defined as the set of vertices $v\in V$ such that $d_G(u, v)\leq R$. We call $R$ the \emph{radius} of $B_G(u,R)$.
The algorithm of Bartal~\cite{Bar96} repeatedly takes a vertex $w$ from the remaining graph, samples a radius $R$ from an exponential distribution with the parameter $\Theta(\log n/\Delta)$, and removes the ball $B_G(w, R)$ from the graph as a single cluster.
Using the memoryless property of the exponential distribution, Bartal~\cite{Bar96} showed that the algorithm yields an $(O(\log n), \Delta)$-LDD.

\paragraph{Low diameter decomposition for bounded pathwidth graphs.}
Abraham et al.~\cite{AGGNT19} showed that if $G$ has pathwidth $\pw$, it admits an $(O(\log \pw),\Delta)$-LDD. While this is a special case of a more general framework for minor-free graphs, we sketch their approach for the bounded pathwidth graphs. 
The following lemma was implicitly proved by several works (e.g.,~\cite{AGGNT19,kamma2017metric}). 
\begin{lemma} \label{lem:path-net}
    Let $G$ be an undirected graph and $P$ be a shortest path in $G$. For any $\Delta>0$, there is a set $N\subset P$ of vertices, called a net, such that 
    \begin{enumerate}  \setlength{\itemsep}{0.1pt}
        \item (Covering) $P\subset \bigcup_{v\in N} B_G(v, \Delta)$, and
        \item (Sparsity) for every $u\in G$, $|B_G(u, 2\Delta) \cap N| = O(1)$. 
    \end{enumerate}
\end{lemma}

Their algorithm works in several rounds. In each round, for each connected component in the current graph, we pick a shortest path $P$ connecting a vertex in the first bag and a vertex in the last bag (in a path decomposition), and compute a net $N$ of that path using Lemma~\ref{lem:path-net}.
Next, for each $v\in N$, we sample a radius $R_v$ from a \emph{truncated exponential distribution} with a parameter $\lambda$ over the range $[\frac{\Delta}{4}, \frac{\Delta}{2}]$ so that each cluster has a diameter at most $\Delta$.
That is, the exponential distribution with the parameter $\lambda$ conditioned on the event that the output belongs to $[\frac{\Delta}{4}, \frac{\Delta}{2}]$. We will define $\lambda$ later.
Finally, we carve $B_G(v,R_v)$ from the current graph as a single cluster.

Crucially, due to the covering property of Lemma~\ref{lem:path-net}, all vertices of $P$ are carved away after the carving step.
Moreover, since $P$ is a path connecting a vertex in the first bag and a vertex in the last bag, $P$ intersects every bag of the path decomposition. Therefore, the pathwidth of the graph decreases by at least one. 
Thus, the algorithm terminates after $\pw$ rounds.
Due to the sparsity property of Lemma~\ref{lem:path-net}, we can show that for every pair $u,v
\in V$, at most $\nu = O(\pw)$ balls, called threateners, affect the probability that $u,v$ belong to different clusters. 

Roughly speaking, the truncated exponential distribution behaves similarly to the exponential distribution, i.e., retaining a memoryless property, up to a small error. These errors from the sampled radii of different threateners are accumulated. In order to keep the total error negligible, the parameter should be at least logarithmic in the number of threateners.
In light of this, Abraham et al.~\cite{AGGNT19} set the parameter $\lambda = \Theta(\frac{\log \nu}{\Delta})=\Theta(\frac{\log \pw}{\Delta})$, yielding an $(O(\log \pw), \Delta)$-LDD of $G$. 
This path-carving framework is the key ingredient that we will adapt to the directed setting later.

\paragraph{Low diameter decomposition for general directed graphs.}
A directed analogue of the low diameter decomposition (Definition~\ref{def:ldd}) was very recently introduced by Bernstein, Gutenberg, and Wulff-Nilsen~\cite{BGW20} (though not stated explicitly), who showed that 
every $n$-vertex directed graph $G$ admits an $(O(\log^2 n), \Delta)$-LDD for every parameter $\Delta>0$. 
Later, Bernstein, Nanongkai, and Wulff-Nilsen~\cite{BNW25} explicitly defined the notion of a low diameter decomposition and presented a simpler algorithm yielding an $(O(\log^2 n), \Delta)$-LDD. 
This bound has been improved to $O(\log n \cdot \log\log n)$ by Bringmann et al.~\cite{BFHL25}. 

The algorithm of Bernstein et al.~\cite{BNW25} follows the ball carving framework: it samples a radius and carves either an in-ball or an out-ball\footnote{We defer the definitions of an in-ball and an out-ball to Section~\ref{sec:overview-ours}} around a chosen vertex.
A key difficulty in the directed setting is that, unlike the undirected case, one cannot guarantee the weak diameter bound of an SCC contained in a single ball.
To resolve this, the algorithm of~\cite{BNW25} recursively applies the same carving procedure to every SCC until all SCCs have a small weak diameter. By carefully selecting centers, Bernstein et al.~\cite{BNW25} managed to bound the recursion depth by $O(\log n)$. This yields an $O(\log^2 n)$ bound on the loss parameter. 

\subsection{Low diameter decomposition for special classes of directed graphs}
The aforementioned issue remains even for special classes of directed graphs, such as those with bounded pathwidth, bounded treewidth, or planar graphs.
We briefly review some prior works that circumvented this issue. We remark that while all prior work primarily addressed random quasipartitions for those graph classes, their quasipartition results directly translate to low diameter decompositions with the same loss parameter\footnote{From~\Cref{def:ldd,def:quasi}, it is easy to check that every $(\beta,\Delta)$-quasipartition is a $(\beta,\Delta)$-LDD.}. Consequently, we present the prior work from the viewpoint of low diameter decompositions.

\paragraph{Separator-based approach.}
M\'emoli, Sidiropoulos, and Sridhar~\cite{memoli2018quasimetric} showed that every $n$-vertex directed graph with treewidth $\tw$ admits an $(O(\tw \log n), \Delta)$-LDD for every parameter $\Delta>0$.
Their approach is based on a balanced separator. 
We call $U\subset V(G)$ a \emph{balanced separator} if  the number of vertices in each connected component of $G\setminus U$ is at most $\frac{2}{3}|V(G)|$.
It is well known that every graph $G$ with treewidth $\tw$ admits a \emph{balanced separator} $U\subset V(G)$ of size at most $\tw+1$. 

The algorithm of~\cite{memoli2018quasimetric} computes a balanced separator $U$, samples a radius, and constructs both in-balls and out-balls around each vertex of $U$. Next, it 
adds the boundary edges of these balls to the cutset $S$. 
Then, it recursively proceeds on each SCC of $G\setminus U$. 
Critically, any SCC of $G\setminus S$ is either disjoint from $U$ or has a weak diameter of at most $\Delta$.
The loss parameter at a single recursion level is $O(\tw)$ since the algorithm computes a cutset through  $2|U|=O(\tw)$ balls. In addition, the recursion depth is $O(\log n)$.
Overall, the algorithm of~\cite{memoli2018quasimetric} yields an $(O(\tw\log n), \Delta)$-LDD.

Kawarabayashi and Sidiropoulos~\cite{kawarabayashi2022embeddings} showed that every $n$-vertex planar directed graph admits an $(O(\log^2 n), \Delta)$-LDD. 
Instead of a (balanced) vertex separator, they utilized the shortest path separator of~\cite{thorup2004compact}, which consists of a constant number of shortest paths.
Again, the $O(\log n)$ factor arises from the fact that the recursion depth is $O(\log n)$.
We remark that it seems difficult to avoid the $O(\log n)$ factor in any separator-based approach since the recursion depth is $O(\log n)$ and the edge-cut probability incurred in each round must be handled independently.

\paragraph{LDD of constant loss parameter for bounded pathwidth directed graphs.} 
Salmasi, Sidiropoulos, and Sridhar~\cite{SSS19} showed that every $n$-vertex directed graph $G$ with pathwidth $\pw$ admits a $(2^{O(\pw^2)}, \Delta)$-LDD for every parameter $\Delta>0$. This is the only known LDD result for structured directed graphs where the loss parameter is independent of $n$.

We briefly sketch their approach at a high level. We start by setting the current graph to $H=G$.
In each round,~\cite{SSS19} computes a $u$-$v$ shortest path $P$ that connects the first and last bags of the path decomposition.
Instead of carving the balls around the net vertices along $P$,
the algorithm computes a random shift $z\in [0,\Delta]$ and adds all edges crossing distances of $i\Delta+z$ from $u$ for every integer $i$, to the cutset. 
Notably, after this, they remove the edge set of $P$ from $H$ (recall that in the ball carving framework, the cutset is removed) and recurse. 
Finally, the algorithm repeats the entire procedure in the reversed direction. 
The recursion depth is $O(\pw^2)$.

Salmasi, Sidiropoulos, and Sridhar~\cite{SSS19} managed to eliminate the dependency on $n$ by carving away the boundary edges of concentric layers with a random shift. 
On the other hand, they remove the edge set of $P$ instead of the cutset edges. This causes a more involved analysis of bounding the weak diameter, and yields a $(2^{O(\pw^2)}, \Delta)$-LDD.
We also remark that their analysis crucially utilizes the structure of the bounded pathwidth graphs, making it seemingly difficult to lift their result to bounded treewidth graphs or planar graphs.

\begin{figure}
    \centering
    \includegraphics[scale=0.7]{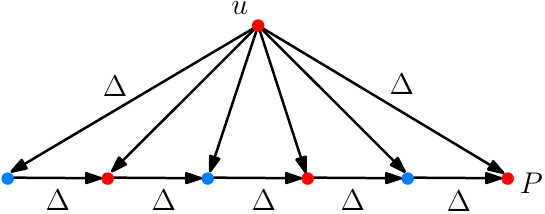}
    \caption{Counter-example of Lemma~\ref{lem:path-net} for directed graphs. Consider a directed graph $G$ consisting of a shortest path $P$ of $n$ edges of weight $\Delta$, an apex vertex $u$, and an edge $(u,p)$ of weight $\Delta$ for every vertex $p\in P$. 
Then, for any subset $N\subseteq P$ that satisfies the covering property, i.e., $P\subset \bigcup_{v\in N} \outb{G}{v,\Delta}$, notice that each ball covers at most $O(1)$ vertices of $P$.
Therefore, $|N| = \Theta(n)$. The blue points indicate the points in $N$.
Since $d_G(u, v) = \Delta$ for every $v\in P$, $N\subseteq \outb{G}{u, 2\Delta}$.}
    \label{fig:densenet}
\end{figure}

\subsection{Our solution: Scissors Carving} \label{sec:overview-ours}
 Our algorithm departs from the previous approaches.
 We develop a new carving procedure, which we call \emph{scissors carving}. This is our main technical contribution that enables us to achieve an $(O(\pw), \Delta)$-LDD.
Let $G=(V,E,w)$ be an edge-weighted directed graph.
For a vertex $u\in V$ and a parameter $R>0$, we define the \emph{out-ball} $\outb{G}{u,R}$ and the \emph{in-ball} $\inb{G}{u,R}$ as
\begin{align*}
    &\outb{G}{u,R} :=\{x\in V\mid d_G(u,x) \leq R\}, \text{ and } 
    \inb{G}{u,R} :=\{x\in V\mid d_G(x,u) \leq R\}, \text{ respectively.}
\end{align*}
We call $R$ the \emph{radius}.

At a high level, suppose we are given a target $x$-$y$ path $P$ that we want to carve away. 
We compute an out-ball around $x$ and an in-ball around $y$ with radii proportional to $d_G(x,y)$, and then remove their boundary edges.
Then, we recurse on each SCC intersecting $P$.

The key insight is that this approach always makes progress for each recursive SCC: either the length of the portion of $P$ contained in each SCC, called the \emph{forward distance}, is reduced, or the pathwidth of the SCC is decreased. 
Once the length of the current path becomes sufficiently small, we apply the same procedure in the reverse direction, i.e., computing an in-ball around $x$ and an out-ball around $y$, to reduce the \emph{backward distance} $d_G(y,x)$.
The separation between the forward-reduction stage and the backward-reduction stage is essential, as reducing the backward distance using two balls is only possible once the forward distance is sufficiently small.

We demonstrate the two advantages of the scissors carving framework as follows.
First, we recursively compute two balls of geometrically decreasing radii. 
Even though we may compute many balls across the recursion, most of them have radii much larger than $\Delta$. This allows us to bound the edge cut probability by $\Pr[e\in S]\leq O(1)\cdot \frac{w(e)}{\Delta}$.

Second, after both the forward-reduction stage and the backward-reduction stage, we ensure that both the forward distance and the backward distance of each remaining SCC $H$ are sufficiently small.
Using this structural advantage, we show that any SCC of $H$ fully contained in a single ball of radius $\Delta$ around $u\in H$ has a weak diameter of at most $O(\Delta)$.
This overcomes the limitation that, in general directed graphs, 
an SCC contained in a single ball of radius $\Delta$ can have an arbitrarily large diameter. 
This allows us to ensure the desired diameter bound.

Based on the scissors carving procedure, we compute an $(O(\pw), \Delta)$-LDD following the framework of~\cite{AGGNT19}. Specifically, we start by setting the current graph to $H=G$. At each stage, we take a shortest path $P$ in $H$ that we want to carve away. In particular, we choose $P$ such that $\pw(H\setminus P) < \pw(H)$.
Then, we apply scissors carving to $P$ and update the cutset $S$. This ensures that every SCC of $H\setminus S$ either has a bounded weak diameter or a reduced pathwidth (See Lemma~\ref{lem:overview-hbc}). Furthermore, we shall guarantee that every SCC disjoint from $P$ has a reduced pathwidth.
Then, we recurse on every SCC of $H\setminus S$ whose weak diameter is greater than $\Delta$.
By taking the union $S$ of the cutsets computed from all recursive calls, we ensure that every SCC of $G\setminus S$ has a weak diameter of at most $\Delta$.
The recursion depth is at most $\pw(G)$. Due to the union bound of probability, the probability bound of Lemma~\ref{lem:overview-hbc} yields $\Pr[e\in S] \leq O(\pw(G))\cdot \frac{w(e)}{\Delta}$.

\begin{lemma} (informal\footnote{Here, we give an informal statement in order to hide several technical terms behind it. See Lemma~\ref{lem:reduce-pathwidth} for details.}) \label{lem:overview-hbc}
    Let $H$ be a strongly connected subgraph of $G$.
    There is a randomized algorithm that computes a cutset $S\subset E(H)$ such that 
    for every strongly connected component $C$ of $H\setminus S$, either $\pw(C) < \pw(H)$ or the weak diameter of $C$ is $O(\Delta)$. Moreover, for every $e\in E(H)$, $\Pr[e\in S]\leq O(1)\cdot \frac{w(e)}{\Delta}$.
\end{lemma}

In the rest of this section,
we explain the core idea of the scissors carving technique and give a proof sketch of Lemma~\ref{lem:reduce-pathwidth}, especially for the first statement. 
We begin by defining some notation used in this paper.
For a vertex set $C\subseteq V$, we denote the set of edges whose source endpoint is in $C$ and target endpoint is outside $C$ by $\delta^+(C)$. Similarly, we denote the set of edges whose source endpoint is outside $C$ and target endpoint is inside $C$ by $\delta^-(C)$.

Let $H$ be a subgraph of $G$ and $P$ be a shortest path in $G$ whose two endpoints are contained in $H$.
For a strongly connected \emph{subgraph} $C$ of $H$, let $u$ and $v$ be the first and the last vertices of $P\cap C$ (along the order of $P$), respectively.
We define the \emph{surviving subpath} of $P$ with respect to $C$, denoted as $P[C]$, as the $u$-$v$ subpath of $P$. Note that $P[C]$ may not be fully contained in $C$.
We define the \emph{forward distance} and the \emph{backward distance} of $C$ (with respect to $P$) as $d_G(u,v)$ and $d_G(v,u)$, respectively, and denote them by $f_P(C)$ and $b_P(C)$, respectively.
Note that we analyze these distances with respect to the shortest paths in $G$, rather than the current subgraph $C$. This is sufficient for our purpose since our goal is to bound the \emph{weak diameter} of each SCC.

We use the following lemma. See Figure~\ref{fig:venn} for an illustration and Section~\ref{sec:ldd} for its proof.
\begin{restatable}{lemma}{twoballfin} \label{lem:twoball-fin}
    Let $B_1=B^*_G(v,R)$ and $B_2=B^{*'}_G(v,R)$ with $*,*'\in\{+,-\}$, and let $C$ be a strongly connected subgraph of $G$. Then, 
    \begin{itemize}  \setlength{\itemsep}{0.1pt}
        \item every SCC of $C \setminus (\delta^*(B_1)\cup \delta^{*'}(B_2))$ is either fully contained in $B_1\cap B_2$ or disjoint from $B_1\cap B_2$, and
        \item  every SCC of $C \setminus (\delta^*(B_1)\cup \delta^{*'}(B_2))$ is either fully contained in $B_1$, or fully contained in $B_2\setminus B_1$, or disjoint from $B_1\cup B_2$.
    \end{itemize}
\end{restatable}

\begin{figure}[t]
	\centering
	\includegraphics[scale=0.45]{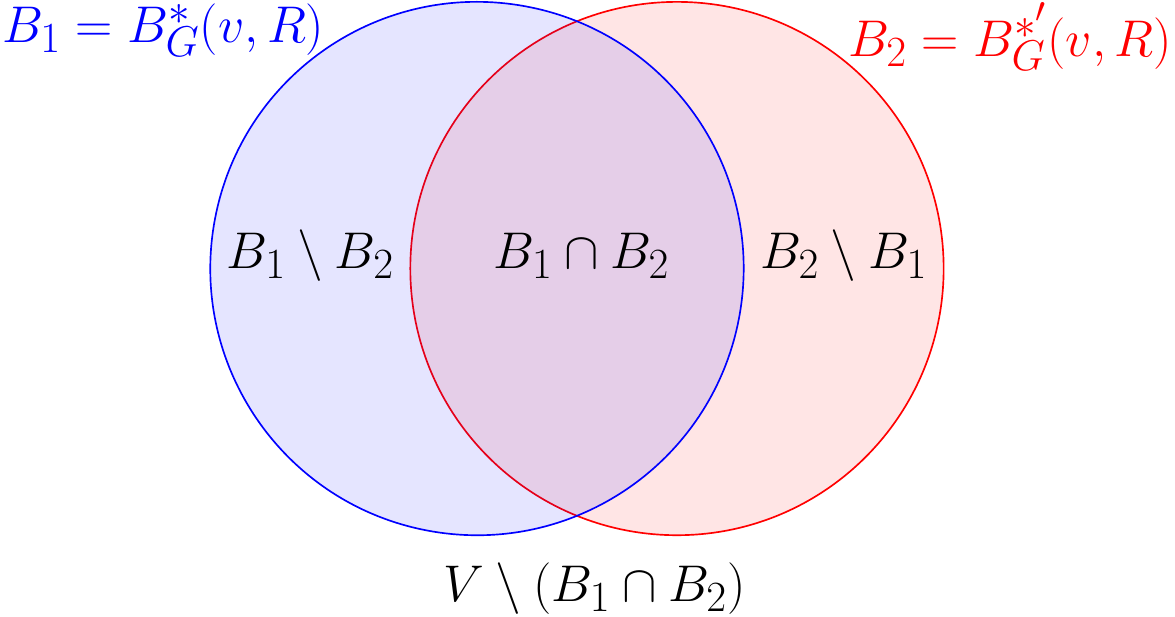}
	\caption{Due to the definition of $\delta^*(\cdot)$, every SCC of $H\setminus \delta^*(B_1)$ is either contained in $B_1$ or disjoint from $B_1$. The same holds symmetrically for $B_2$.
    Therefore, every SCC of $H\setminus (\delta^*(B_1)\cup \delta^{*'}(B_2))$ is fully contained in one of the regions of the Venn diagram. This proves both items of Lemma~\ref{lem:twoball-fin} for the special case of $C=G$. See Section~\ref{sec:ldd} for the general case.
    } 
	\label{fig:venn}
\end{figure}

Next, we explain the scissors carving procedure. 
We are given a strongly connected subgraph $H$ of $G$ and a shortest path $P$ as input, such that the two endpoints of $P$ are contained in $H$, and $\pw(H\setminus P)<\pw(H)$.  
The algorithm consists of three stages, and the first two stages operate recursively.
As we will see below, these stages aim to reduce the forward distance, reduce the backward distance, and bound the weak diameter of each SCC, respectively.
\paragraph{First stage: reducing forward distance.}
In the first stage we are given a strongly connected subgraph $C$ of $H$ as input. Let $P[C]=\langle u,\ldots, v\rangle$ and $\tau = f_P(C)$. 
We sample the radius $R$ uniformly at random from $[\frac{1}{2}\tau, \frac{2}{3}\tau]$, and compute two balls $B_1=\outb{G}{u,R}$, $B_2=\inb{G}{v,R}$ and add their boundary edges $\delta^+(B_1)\cup \delta^-(B_2)$ into the cutset $S$. 
We then recursively apply this approach to each SCC of $C\setminus S$ contained in $B_1\cup B_2$. 
We halt the first stage when $f_P(C)\leq \Delta$, and proceed to the second stage.
This completes the algorithm for the first stage.

We show that the above procedure ensures that each resulting SCC for the recursive calls has a reduced forward distance, while the remaining SCCs have a reduced pathwidth. Hence, it suffices to proceed with the SCCs of the former type, i.e., those with a reduced forward distance.
To see this, since we take $R$ from the range $[\frac{1}{2}\tau, \frac{2}{3}\tau]$ and $P[C]$ is a shortest path in $G$, for every $w\in P[C]$, either $d_G(u, w)\leq R$ or $d_G(w,v)\leq R$. Hence, $P\subseteq B_1\cup B_2$.
In addition, due to Lemma~\ref{lem:twoball-fin}, every SCC of $C\setminus S$ is either fully contained in $B_1$, or fully contained in $B_2\setminus B_1$, or disjoint from $B_1\cup B_2$.
Let $C'$ be an SCC of $C\setminus S$ and $P[C']=\langle x,\ldots, y\rangle$. Since $P[C']$ is a subpath of $P$ and $C'\subseteq C$, $P[C']$ is a subpath of $P[C]$.
If $C'$ is contained in $B_1$, we have that $$d_G(x,y) \leq d_G(u, y) \stackrel{B_1}\leq R \leq \frac{2}{3}\tau.$$ Otherwise, if $C'$ is contained in $B_2\setminus B_1$, we have that $$d_G(x,y)\leq d_G(x,v) \stackrel{B_2}\leq R \leq \frac{2}{3}\tau.$$
Finally, if $C'$ is disjoint from $B_1\cup B_2$, $P\subseteq B_1\cup B_2$ implies that $C'$ is a subset of $C\setminus P$. Then, we have that $\pw(C') \leq \pw(C\setminus P) \leq \pw(H\setminus P) < \pw(H)$.
Hence, we conclude that every SCC of $G\setminus S$ either has a reduced forward distance or a reduced pathwidth, as desired.

\paragraph{Second stage: reducing backward distance.}
In the second stage, we are given a strongly connected subgraph $C$ of $H$ as input with the additional property that $f_P(C)\leq \Delta$. 
Let $P[C]=\langle u,\ldots, v\rangle$ and $b_P(C)=\tau'$.
Importantly, throughout this stage, we ensure that the condition on $f_P(C)$ always holds.
Our algorithm is almost identical to that of the first stage, except that we compute two balls in the reversed direction and sample the radius $R$ from a narrower interval.
The reason we choose a narrower interval is to reduce the backward distance by a constant fraction.

More precisely, we sample $R$ uniformly at random from $[\frac{1}{2}\tau', \frac{13}{24}\tau']$, compute two balls $B_1=\inb{G}{u,R},B_2=\outb{G}{v,R}$ and add $\delta^-(B_1)\cup \delta^+(B_2)$ to the cutset $S$.
We recursively apply this approach to each SCC of $C\setminus S$ contained in $B_1\cup B_2$. We halt the second stage when $b_P(C)\leq 8\Delta$. Here, $8$ is also a carefully chosen constant in order to bound the backward distance.

We show that a symmetrical property holds for this stage: Every SCC contained in $B_1\cup B_2$ has a reduced backward distance, and every SCC disjoint from $B_1\cup B_2$ has a reduced pathwidth.
Let $C'$ be an SCC of $C\setminus S$ and $P[C']=\langle x,\ldots, y\rangle$.
First, we bound $b_P(C')$ in the case where $C'\subseteq B_1\cup B_2$.
From Lemma~\ref{lem:twoball-fin}, without loss of generality, let us assume $C'\subseteq B_1$. The other case where $C'\subseteq B_2\setminus B_1$ is analogous.
From the argument of the first stage, $P[C']$ is a subpath of $P[C]$ and $d_G(u,v) = f_P(C) \leq \Delta$. In addition, due to the triangle inequality, we have that:
$$d_G(y,x) \stackrel{\triangle}\leq d_G(y, u) + d_G(u,x) \leq d_G(y,u) + d_G(u,v) \stackrel{B_1}\leq R + d_G(u,v) \leq R + \Delta \leq \frac{13}{24}\tau' + \Delta.$$
Since the halting condition is $\tau' \leq 8\Delta$, we have that $\frac{13}{24}\tau' + \Delta \leq \frac{2}{3}\tau'$ throughout the second stage. Hence, the backward distance is reduced by $2/3$.

Next, we bound the pathwidth of $C'$ in the case where $C'$ is disjoint from $B_1\cup B_2$.
The difference is that $P[C]$ may not be fully contained in $B_1\cup B_2$.
For simplicity, let us assume that $u$ and $v$ are the leftmost and rightmost vertices of $P$ in the path decomposition, respectively. That is, for two bags $X_u$ and $X_v$ containing $u$ and $v$ respectively, every bag $X$ that intersects $P$, lies between $X_u$ and $X_v$ in the path decomposition\footnote{Indeed, we can guarantee this assumption by carefully taking certain subpaths in each recursion step. See Section~\ref{sec:ldd} for details.}.
We consider a shortest $v$-$u$ path $P'$ of length $\tau'$ in $G$.
Since $R\geq \frac{1}{2}\tau'$, it is easy to check that $P'\subseteq B_1\cup B_2$.
Moreover, since $P'$ intersects every bag $X$ between $X_u$ and $X_v$,
$P'$ intersects every bag $X$ that intersects $P$.
Since $C'$ is disjoint from $B_1\cup B_2$, it is disjoint from $P'$ as well.
Hence, for each bag $X$ intersecting $P\cap C'$, we have that:
$$|C'\cap X| < |C'\cap X| + |P'\cap X| = |(C'\cup P')\cap X| \leq |C\cap X|.$$
This establishes that $\pw(C') < \pw(C)$, as desired.
See Figure~\ref{fig:overview-hbc}(b) for an illustration.

\paragraph{Third stage: reducing the weak diameter.}
Finally, the third stage serves as the base case. We are given a strongly connected subgraph $C$ of $H$ as input, where the two previous stages guarantee that $f_P(C)\leq \Delta$ and $b_P(C)\leq 8\Delta$.
One might ask whether $f_P(C), b_P(C) = O(\Delta)$ directly implies that the weak diameter of $C$ is $O(\Delta)$. This is not true in general. What we can guarantee is only the (bi-)distance between two endpoints of $P[C]$, which does not necessarily bound the distance between every pair of vertices in $C$. See Figure~\ref{fig:third-stage} for an illustration.

Let $P[C]=\langle u,\ldots, v\rangle$. 
We compute two balls $\outb{G}{u, R_1}$ and $\inb{G}{u,R_2}$ around the same vertex $u$, where we sample $R_1,R_2$ uniformly at random from $[\Delta, 2\Delta]$ and $[9\Delta, 10\Delta]$, respectively. Then, we add $\delta^+(B_1)\cup \delta^-(B_2)$ into the cutset $S$.
Every SCC contained in $B_1\cap B_2$ has a small weak diameter, since for every pair of vertices in $B_1\cap B_2$, we can draw a detour of length $O(\Delta)$ through $u$.
In addition, it is straightforward to verify that $P[C]\subset B_1\cap B_2$. Subsequently, every SCC disjoint from $B_1\cap B_2$ has a reduced pathwidth. Due to Lemma~\ref{lem:twoball-fin} we conclude that every SCC of $G\setminus S$ after the third stage either has a weak diameter of at most $\Delta$ or a reduced pathwidth. 
This completes the proof sketch of Lemma~\ref{lem:overview-hbc}.

\begin{figure}
    \centering
    \includegraphics[scale=0.75]{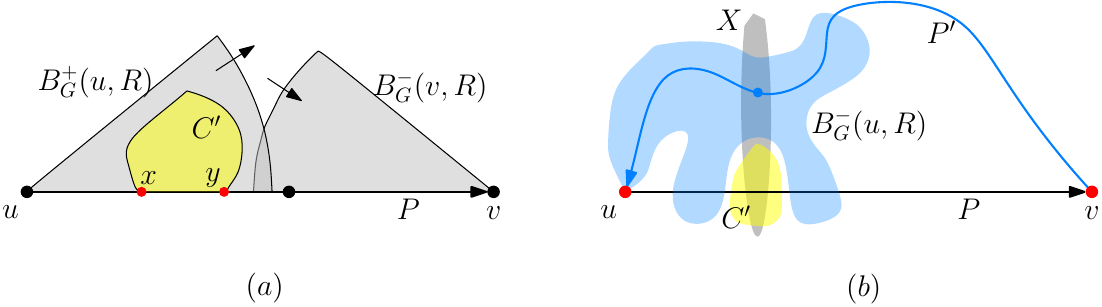}
    \caption{$(a)$ Illustration of the first stage. The gray region denotes two balls $\outb{G}{u,R}$ and $\inb{G}{v,R}$ where $R$ is sampled from $[\frac{1}{2}\tau, \frac{2}{3}\tau]$ with $\tau = d_G(u,v)$. For an SCC $C'$ (yellow region) contained in $B_1=\outb{G}{u, R}$  and $x,y\in P\cap C'$, we have that $d_G(x,y) \leq d_G(u,y) \leq \frac{2}{3}\tau$, since $P$ is a shortest path in $G$. $(b)$ Illustration of the second stage. The blue region denotes $\inb{G}{u, R}$, and $P'$ represents the shortest $v$-$u$ path in $G$. In addition, the yellow region denotes an SCC $C'$ which is disjoint from $\inb{G}{u,R}$.
    Let us assume that $u$ and $v$ are the leftmost and rightmost vertices of $P$ with respect to the path decomposition.
    Then, for every bag $X$ that intersects $P\cap C'$, $X$ also intersects $P'$ (blue vertex). Since $C'$ is disjoint from $P'$, we can derive $\pw(C') <\pw(C)$.} 
    \label{fig:overview-hbc}
\end{figure}

\paragraph{Remark.}
One may expect a logarithmic loss parameter of $O(\log(\ppw(G)))$ as in~\cite{AGGNT19}. However, our scissors carving framework differs from previous ball carving in the following sense. In previous work, once the interior of a ball was carved away, the resulting cluster became ``safe", meaning that it was no longer affected by further ball carving steps. 
This allowed~\cite{AGGNT19} to sample the radius from a (truncated) exponential distribution and utilize its (almost) memoryless property, ensuring them to bound the loss parameter by the logarithmic number of threateners.

In contrast, in our scissors carving, once we compute two balls that cover a shortest path, the interiors of these balls are conveyed to the next round. That is, each edge can be threatened multiple times regardless of whether it is inside a ball or not. 
Thus, the cut probability becomes linear in the number of threateners. 
We also remark that the same phenomenon is observed in the directed treewidth case~\cite{memoli2018quasimetric}.

\begin{figure}
    \centering
    \includegraphics[scale=0.75]{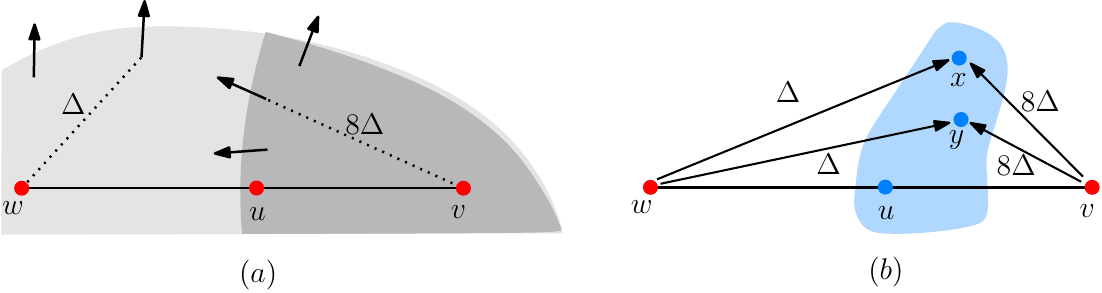}
    \caption{$(a)$ Illustration of the graph partition from the first two stages. In the first stage, an out-ball $\outb{G}{w,\Delta}$ (light gray region) containing $u,v$ is computed. In the second stage, an out-ball $\outb{G}{v, 8\Delta}$ (dark gray region) containing $u$ is computed. $(b)$ At the start of the third stage, an SCC (blue region) might have a large diameter, as no upper bounds for $d_G(x,u)$ and $d_G(y,u)$ are assumed. In the third stage, we reduce the weak diameter of an SCC intersecting $P$ by computing two balls of radii $O(\Delta)$ around $u$.}
    \label{fig:third-stage}
\end{figure}

\section{Directed LDDs for Bounded-Pathwidth Directed Graphs} \label{sec:ldd}
\pathwidthldd*
In this section, we prove Theorem~\ref{thm:ldd}.  
While we state our result in terms of pathwidth, we find it more convenient to use the notion of \emph{path partition-width} (see Definition~\ref{def:pp} below). 
This is a directed analogue of the (undirected) tree partition-width~\cite{ding1996tree} for the directed graphs with small pathwidth. 
\begin{definition}[Path Partition] \label{def:pp}
    A \emph{path partition} of a directed graph $G=(V,E)$ is a path $\mathcal P$ whose vertices are bijectively associated with the sets of a partition $\mathcal X=\{X_1,X_2,\ldots, X_m\}$ of $V$, called \emph{bags}, such that for every $(u,v)\in E$, there exist $X_i, X_{i+1}\in \mathcal X$ for some $1\leq i\leq m-1$ such that $\{u,v\}\subset X_i\cup X_{i+1}$. The \emph{width} of $\mathcal X$ is $\max_{i\in [m]}\{|X_i|\}$. The path partition-width $\ppw(G)$ of $G$ is the minimum width among all path partitions of $G$.
\end{definition}

To begin with, we apply an edge-preserving isometric embedding from a directed graph with bounded pathwidth into another directed graph with bounded path partition-width. 
Crucially, this embedding preserves the edge structure of the graph.
This will allow us to lift an LDD of the embedded graph back to the original graph.
Note that our embedding is not new. Filtser et al.~\cite{filtser2025optimal} defined a similar embedding from (undirected) bounded treewidth graphs into graphs of bounded tree partition-width\footnote{Tree partition can be viewed as a natural generalization of path partition, in the same way that tree decomposition generalizes path decomposition.} graphs. Although not stated explicitly, their embedding is also an edge-preserving isometric embedding.

\begin{restatable}{definition}{epie}\textup{(Edge-preserving isometric embedding)}\textup{\cite{filtser2025optimal}} \label{def:epie} 
    Let $G=(V,E,w)$, $H=(U, E_H,w_H)$ be directed graphs. We say that a map $\phi:V\rightarrow U$ is an \emph{edge preserving isometric embedding} if the following conditions hold:
\begin{itemize}
    \item (isometry) $d_G(u,v)=d_H(\phi(u),\phi(v))$ for every $u,v\in V$, 
    \item (edge preservation) there is an injective mapping $\bar \phi: E\rightarrow E_H$ such that for every $e\in E$, $w(e)=w_H(\bar\phi(e))$ and $\bar\phi(e)$ belongs to a shortest $\phi(u)$-$\phi(v)$ path in $H$. 
\end{itemize}
\end{restatable}

While the ultimate goal is to design an LDD for directed graphs of bounded pathwidth, our construction algorithm is formally defined on graphs of bounded path partition-width.  
We remark that a similar approach was used by Filtser et al.~\cite{filtser2025optimal} to compute an LDD for (undirected) graphs of bounded treewidth.

We show that there is an edge-preserving isometric embedding from graphs with bounded pathwidth into graphs with bounded path partition-width.
\begin{restatable}{lemma}{embedwidth} \label{lem:embed-width}
    Let $G=(V,E,w)$ be a directed graph. Then, there is an edge-preserving isometric embedding from $G$ into a directed graph $H=(U,E_H,w_H)$ such that $\ppw(H) \leq \pw(G)+1$.
\end{restatable}
\begin{proof}
    The proof of this lemma follows the same line of reasoning as Lemma 2.2 of~\cite{filtser2025optimal}, but we state the full proof for completeness.
    Let $(P, \mathcal X)$ be a path decomposition of $G$ with width $k=\pw(G)$.
    We compute $H$ and its path partition as follows. For each $v\in V(G)$, if it appears in $t$ consecutive bags $X_1,X_2,\ldots, X_t$ of $\mathcal X$, we make $t$ copies $v_1,v_2,\ldots, v_t$ of $v$ and replace each vertex $v\in X_i$ with the $i$-th copy $v_i$. We then set $\phi(v)=v_1$. This defines the vertex set $U$ and the embedding $\phi$.

    If $X_i$ and $X_j$ are two adjacent bags, i.e., $|i-j|=1$, we connect $v_i$ and $v_j$ with two bidirectional edges of zero weight. For each edge $(u,v)\in E(G)$, there is a bag $X_i\in \mathcal X$ that contains both $u$ and $v$. If there are multiple such bags, we choose an arbitrary one. We add an edge $(u_i, v_i)$ with weight $w(u,v)$.
    This completes the construction of $E_H$ and $w_H$.

    We show that $\phi$ is an edge-preserving isometric embedding. The isometry property follows from the fact that any two adjacent copies of $v$ are connected by zero-weight bidirectional edges, and for each edge $(u,v)\in E$, there is a unique corresponding edge $(u', v')$ with weight
    $w_H(u',v')=w(e)$. 
    For the edge preservation property, let us consider an edge $e=(u,v)$ in $G$, and
    let $\bar \phi(e)=(u',v')$ be the corresponding edge in $E_H$.
    By construction, it is obvious that the resulting map $\bar\phi$ is injective.
    Note that the set of bags of $(P, \mathcal X)$ containing $v$ forms a subpath of $P$.
    Hence, there is a $v'$-$\phi(v)$ path $\pi_v$ of zero length in $H$. 
    Similarly, there is a $\phi(u)$-$u'$ path $\pi_u$ of zero length in $H$. 
    Finally, let $\pi_{u,v}$ be the concatenation of $\pi_u$, $(u',v')$, and $\pi_v$. 
    Then $\pi_{u,v}$ consists of exactly one edge of weight $w(e)$ and other edges of zero weight.
    The length of $\pi_{u,v}$ is therefore $w(e)$.
    Due to the isometry property, and since $d_G(u,v) \leq w(e)$, we conclude that $\pi_{u,v}$ is a shortest $\phi(u)$-$\phi(v)$ path containing $\bar \phi(e)$.
    This establishes the edge preservation property.

    Finally, we compute the path partition $(P, \mathcal X')$ of $H$ using $(P, \mathcal X)$.
    More specifically, we use the same path $P$, and 
    for each bag $X\in \mathcal X$, we compute a bag $X'$ consisting of the copies of the vertices of $X$. Obviously, the resulting collection $\mathcal X':=\{X'\mid X\in \mathcal X\}$ is a partition of $U$.
    Since the width of $(P, \mathcal X)$ is $k$, the maximum bag size is $\max_{X\in \mathcal X}|X| = k+1$.
    Consequently, 
    $$\max_{X'\in \mathcal X'}|X'|=\max_{X\in \mathcal X}|X| = \pw(G)+1.$$
    Therefore, $(P, \mathcal X')$ is a path partition of width $\pw(G)+1$ in $H$, which implies $\ppw(H)\leq \pw(G)+1$. This completes the proof. 
\end{proof}

Crucially, an LDD of the embedded graph can be transferred back to the original graph without any additional loss. 
\begin{restatable}{lemma}{embedLDD} \label{lem:embed-LDD}
    Let $\phi$ be an edge-preserving isometric embedding from $G$ into $H$, and
    $\mathcal D$ be a $(\beta, \Delta)$-LDD of $H$. Then there is a $(\beta, \Delta)$-LDD of $G$.
\end{restatable}
\begin{proof}
    We obtain a distribution $\mathcal D'$ over the cutsets of $G$ as follows.
    We sample a cutset $S\sim \mathcal D$. Note that $\Pr[e\in S]\leq \beta\cdot\frac{w(e)}{\Delta}$ implies that $\Pr[e\in S]$ is zero for each edge $e$ of zero weight.
    Therefore, each edge $e\in S$ has positive weight. 
    Let $S':=\{\bar \phi^{-1}(e)\mid e\in S\}$ be the cutset of $G$. Here, we say $e'=\bar\phi^{-1}(e)$ if $e = \bar\phi(e')$. 
    We show that the induced distribution of $S'$ is indeed a $(\beta, \Delta)$-LDD of $G$.

    Let $C$ be an SCC of $G\setminus S'$. First, we show that the weak diameter of $C$ is at most $\Delta$.
    For each edge $e=(u,v)$ in $C$, 
    $e\notin S'$ implies $\bar \phi(e)\notin S$.
    Due to the edge preservation property, there is a shortest $\phi(u)$-$\phi(v)$ path $\pi_{u,v}$ in $H$, one of whose edges is $\bar \phi(e)$.
    Due to the isometry property, all other edges of $\pi_{u,v}$ have zero weight.
    Since $\bar\phi(e)\notin S$ and $S$ does not contain a zero-weight edge, $\pi_{u,v}$ is disjoint from $S$. Hence, we have that $\pi_{u,v}$ is contained in $H\setminus S$.
    Since $C$ is an SCC of $G\setminus S'$, the set $C':=\{ \phi(v)\mid v\in C\}$ is strongly connected in $H\setminus S$. 
    Since $C'$ is a subset of an SCC of $H\setminus S$ and $S$ is a random cutset that induces a $(\beta, \Delta)$-LDD of $H$, the weak diameter of $C'$ in $H$ is at most $\Delta$.
    Then, due to the isometry property, the weak diameter of $C$ in $G$ is equal to the weak diameter of $C'$ in $H$.
    Hence, the weak diameter of $C$ is at most $\Delta$, as desired.
    
    Second, we analyze the probability $\Pr[e\in S']$. From the definition of $S'$, 
    \begin{align*}
        \Pr[e\in S'] &= \Pr[\bar \phi(e)\in S] \leq \beta\cdot \frac{w_H(\bar \phi(e))}{\Delta}
        = \beta\cdot \frac{w_G(e)}{\Delta}.
    \end{align*}
    Here, the last equality comes from the edge preservation property.
    This completes the proof.    
\end{proof}

Based on~\Cref{lem:embed-LDD,lem:embed-width}, it suffices to show the existence of $(O(k), \Delta)$-LDD for the directed graphs with path partition-width at most $k$. 
Throughout this section, we let $G$ be a directed graph, and let $(\mathcal P, \mathcal X=\{X_1,X_2,\ldots, X_m\})$ be the path partition of $G$ with width $\ppw(G)$.
If it is clear from the context, we slightly abuse the notation and simply denote the path partition $(\mathcal P, \mathcal X)$ by $\mathcal P$.
Furthermore, we may assume that $G$ is strongly connected as otherwise, we can independently execute our algorithm on each SCC of $G$.
For a subgraph $H$ of $G$, we define the \emph{induced path partition} of $H$ with respect to $\mathcal P$ as $(\mathcal P, \{X_i\cap V(H)\mid X_i\in \mathcal X\})$, and denote it by $\mathcal P[H]$.
We denote the width of $\mathcal P[H]$ by $w(\mathcal P[H])$. Note that for every subgraph $H$ of $G$, $\ipw{H}\leq \ipw{G} = \ppw(G)$ trivially holds.

The remainder of this section is devoted to proving the following lemma.
\begin{lemma} \label{lem:ldd-pathpartition}
    For any $\Delta>0$, there is an $(O(\ppw(G)),\Delta)$-LDD of $G$.
\end{lemma}
Before proceeding to the main proof, we first establish the following auxiliary lemma, which will be useful throughout our analysis.

\twoballfin* 
\begin{proof}
    We claim that every SCC of $C\setminus (\delta^*(B_1)\cup \delta^{*'}(B_2))$ is either fully contained in $B_1$ or disjoint from $B_1$.
    Without loss of generality, let $*=*'=``+"$.
    Suppose by contradiction that there exists an SCC $C'$ of $C\setminus (\delta^+(B_1)\cup \delta^+(B_2))$ that contains both a vertex $x\in B_1$ and a vertex $y\notin B_1$. Let $P$ be an $x$-$y$ path fully contained in $C$.
    Since $x\in B_1$ and $y\notin B_1$, we have that $d_G(u,x) \leq R$ and $d_G(u,y) > R$.
    Therefore, there must exist an edge $(x',y')\in P$ such that $d_G(u,x') \leq R$ and $d_G(u,y')>R$. Then, by the definition of $\delta^+(B_1)$, $(x',y')\in \delta^+(B_1)$.
    This contradicts the assumption that $P$ is fully contained in $C'\subseteq C\setminus (\delta^+(B_1)\cup \delta^+(B_2))$.
    Therefore, every SCC is either fully contained in $B_1$ or disjoint from $B_1$. 
    This proves the claim.
    By symmetric arguments, the same property holds for $B_2$: every SCC of $C\setminus (\delta^*(B_1)\cup \delta^{*'}(B_2))$ is either fully contained in $B_2$ or disjoint from $B_2$.

    Combining these two claims, we conclude that every SCC is either fully contained in $B_1 \cap B_2$, or contained in $B_1 \setminus B_2$, or contained in $B_2 \setminus B_1$, or disjoint from $B_1 \cup B_2$.
    This immediately establishes both statements of the lemma.
\end{proof}

Next, we present a randomized recursive algorithm $\ld$, which receives as input a directed graph $G=(V,E,w)$, a path partition $\mathcal P$ of $G$, a strongly connected subgraph $H$ of $G$, and a parameter $\Delta>0$. In addition, $\ld$ computes a (randomized) cutset of $H$ whose distribution induces a $(O(\ipw{H}), \Delta)$-LDD of $H$. See~\Cref{alg:dir-ldd} for the pseudocode.
The algorithm immediately terminates if $\ipw{H}=0$, which implies that $H$ is an empty subgraph. 
For the other case where $\ipw{H}>0$,~\Cref{alg:dir-ldd} calls \rp algorithm\footnote{The algorithm receives $\bar P$ as an additional input parameter, where $\bar P$ is a sequence of vertices in $G$ satisfying certain conditions. We provide the details of $\bar P$ in Section~\ref{sec:dpp}.}, whose function is summarized in the following lemma. We defer the detailed description of \rp and the proof of Lemma~\ref{lem:reduce-pathwidth} to Sections~\ref{sec:dpp} and~\ref{sec:LDD-analysis}, respectively.

\begin{restatable}{lemma}{reducepathwidth} \label{lem:reduce-pathwidth} 
    There is an algorithm $\rpp(G,\mathcal P,H, \bar P, \Delta)$ which computes a (random) cutset $S$ of $E(H)$ such that:
    \begin{itemize}
    \item[1.] For every strongly connected component $C$ of $H\setminus S$,
    either $\ipw{C} < \ipw{H}$ or the weak diameter of $C$ is at most $12\Delta$,
    and
    \item[2.] for every edge $e$ in $E(H)$, $\Pr[e\in S]\leq O(1)\cdot \frac{w(e)}{\Delta}$.
    \end{itemize}
\end{restatable}

Then, for each strongly connected component $C$ of $H\setminus S$ such that the weak diameter of $C$ is greater than $\Delta$, the algorithm executes a recursive call $\ldd(G,\mathcal P,C,\Delta)$ to compute a random cutset $S_C$ of $C$.
Finally, $\ld$ returns the union of $S$ and all cutsets $S_C$ returned by the recursive calls. This completes the description of $\ld$.

\begin{algorithm}[h] 
    \caption{\textsf{Directed-LDD}$(G, \mathcal P, H, \Delta)$}
	\begin{algorithmic}[1] 
        \If{$\ipw{H}=0$} 
            \State
            \Return $\emptyset$
        \EndIf
        \State $S\gets$ \rpp$(G,\mathcal P,H, \bar P, \Delta)$
        \State $\mathcal C\gets$ Collection of SCCs of $H\setminus S$ whose weak diameter is greater than $12\Delta$
        \ForAll{$C\in \mathcal C$} 
            \State $S_C\gets$ $\ld(G,\mathcal P, C,\Delta)$
            \State $S\gets S\cup S_C$
        \EndFor
        \State \Return $S$
        \end{algorithmic}
	\label{alg:dir-ldd}
\end{algorithm}

\paragraph{Analysis of \textnormal{$\ld$}.}
In this part, we prove~\Cref{lem:ldd-pathpartition} assuming Lemma~\ref{lem:reduce-pathwidth}. 
Let $S=\ld(G,\mathcal P, H, \Delta)$. First, we bound the weak diameter of every SCC in $H\setminus S$.
\begin{lemma} \label{lem:ldd-diameter}
    The weak diameter of each SCC of $H\setminus S$ is at most $12\Delta$.
\end{lemma}
\begin{proof}
    We use induction on $\ipw{H}$. The base case where $\ipw{H}=0$ is immediate.
    For the inductive step, suppose $\ipw{H}=i$.
    Due to Lemma~\ref{lem:reduce-pathwidth}, after line 3 of $\ld$, every SCC of $H\setminus S$ either belongs to $\mathcal C$ or has a weak diameter of at most $12\Delta$. 
    First, let $C$ be an SCC in the former case.
    Due to Lemma~\ref{lem:reduce-pathwidth}, $\ipw{C}\leq i-1$.
    Hence, by the induction hypothesis, after line 6, every SCC of $C\setminus S_C$ has a weak diameter of at most $12\Delta$. 
    Since $S$ becomes a superset of $S_C$ after line 7, every SCC of $C\setminus S$ consequently has a weak diameter of at most $12\Delta$.
    
    Next, let $C'$ be an SCC in the latter case.
    Recall that the weak diameter of a subgraph is always upper-bounded by that of its supergraph. 
    Hence, we conclude every SCC of $C'\setminus S$ has a weak diameter of at most that of $C'$, which is bounded by $12\Delta$. 
    By combining the two cases, we conclude that every SCC of $H\setminus S$ has a weak diameter of at most $12\Delta$. This completes the proof.
\end{proof}

Next, we analyze the edge-cut probability $\Pr[e\in S]$ for an edge $e\in E(H)$.
\begin{lemma} \label{lem:ldd-prob}
    For every edge $e$ in $H$, $\Pr[e\in S] \leq O(\ipw{H}))\cdot\frac{w(e)}{\Delta}$.
\end{lemma}
\begin{proof}
    The recursion depth of $\ld(G, \mathcal P ,H, \Delta)$ is at most $\ipw{H}$ since $\ipw{H}$ decreases at each recursion level by Lemma~\ref{lem:reduce-pathwidth}. 
    In addition, since each edge $e$ is contained in at most one SCC in $\mathcal C$, $e$ belongs to at most one recursive call executed in line 6.
    Therefore, $e$ is considered in at most $\ipw{H}$ executions of $\rpp$ in line 3. Due to Lemma~\ref{lem:reduce-pathwidth} and the union bound of probability, $\Pr[e\in S]\leq O(\ipw{H}\cdot \frac{w(e)}{\Delta}$, as desired. 
\end{proof}
Finally, \Cref{lem:ldd-pathpartition} follows from applying~\Cref{lem:ldd-diameter,lem:ldd-prob} by taking $H\leftarrow G$ and $\Delta \leftarrow \frac{\Delta}{12}$.

\subsection{Reducing path partition-width} \label{sec:dpp}
In this section, we explain our key subroutine \rpp.
To this end, we define some additional notation that will be used in this section.
Recall that $\mathcal P=(\mathcal P, \mathcal X)$ consists of $m$ bags $X_1,X_2,\ldots X_m$.
Let $H$ be a strongly connected subgraph of $G$.
We say an index $i\in [m]$ is a \emph{critical} index of $H$ if $|V(H)\cap X_i|=\ipw{H}$.
Equivalently, we refer to $X_i$ as a \emph{critical bag} of $H$.
Let $\mathcal C(H)=\{j_1,j_2,\ldots, j_\ell\}$ be the set of critical indices of $H$ such that $1\leq j_1 <j_2<\ldots< j_\ell\leq m$.

We say a shortest $u$-$v$ path $P$ in $G$ is a \emph{critical path} with respect to $H$ if 
$u\in X_{j_1}\cap H$ and $v\in X_{j_\ell}\cap H$.
From the definition of a path partition, $P$ intersects every bag $X_t$ with $j_1\leq t\leq j_\ell$. Hence, $P$ intersects every critical bag of $H$, which leads to the following observation.
\begin{observation} \label{obs:critical-path}
    For any critical path $P$ with respect to $H$, $\ipw{H\setminus P} < \ipw{H}$.
\end{observation}

Due to technical reasons, throughout the algorithm, we maintain a certain subsequence of the critical path $P$, called a \emph{critical sequence}, which is defined as follows.
For each critical index $j$ of $H$, let $v_j\in X_j\cap H$ be the vertex that appears earliest in $P$ along the direction from $u$ to $v$. We may assume $v=v_{j_\ell}$ as otherwise, we can take the $u$-$v_{j_\ell}$ subpath of $P$. 
Then, $v_{j_1}, v_{j_2},\ldots, v_{j_\ell}$ are naturally ordered along the path $P$.
We call $\langle v_{j_1},v_{j_2},\ldots v_{j_\ell}\rangle$ a \emph{critical sequence} of $H$ with respect to $P$, and denote it by $\bar P[H]$. Note that this critical sequence $\bar P[H]$ is uniquely defined for a fixed $H$ and $P$. See Figure~\ref{fig:critical-sequence}(a) for an illustration.
Since $\bar P[H]$ intersects every critical bag of $H$, we observe the following.
\begin{observation} \label{obs:critical-sequence}
    For a critical sequence $\bar P[H]$ of $H$, $\ipw{H\setminus \bar P[H]} < \ipw{H}$.
\end{observation}

For a critical sequence $\bar P=\langle u,\ldots, v\rangle$ of $H$ with respect to $P$, and a strongly connected subgraph $C$ of $H$, let $x$ and $y$ be the first and the last vertices of $\bar P\cap C$, respectively.
We define the \emph{surviving subsequence} of $\bar P$ with respect to $C$, denoted by $\bar P[C]$, as the consecutive subsequence of $\bar P[H]$ that starts at $x$ and ends at $y$.
Note that while $\bar P[C]$ may not be fully contained in $C$, the two endpoints $x,y$ are contained in $C$. See Figure~\ref{fig:critical-sequence}(b) for an illustration.

We define the \emph{forward disatnce} and the \emph{backward distnace} of $C$ (with respect to $P$) as $d_G(x,y)$ and $d_G(y,x)$, respectively, and denote them by $f_{\bar P}(C)$ and $b_{\bar P}(C)$, respectively. 
Note that we analyze these distances with respect to the shortest paths in the original graph $G$, rather than the subgraph $C$. This is sufficient for our purpose since our goal is to bound the \emph{weak diameter} of each SCC. We show that $\bar P[C]$ plays the role of a ``critical" sequence for $C$. That is, the removal of $\bar P[C]$ from $C$ reduces its path partition-width. Based on the following lemma, in the remainder of this section, we slightly abuse the notation and refer to $\bar P[C]$ as a \emph{critical sequence} of $C$. 
\begin{lemma}\label{lem:critical-path} 
    Let $C$ be a strongly connected subgraph of $H$ and let $\bar P$ be a critical sequence of $H$ with respect to a critical path $P$.
    Then, $\ipw{C\setminus \bar P[C]} < \ipw{H}$.
\end{lemma}
\begin{proof}
    We show that for every critical bag $X$ of $H$, $|X\cap (C\setminus \bar P[C])| < |X\cap H|$ holds.
    Assume to the contrary that there is a critical bag $X$ satisfying $|X\cap (C\setminus \bar P[C])| \geq |X\cap H|$. Since $H$ is a supergraph of $C\setminus \bar P[C]$, the equality must hold. By Observation~\ref{obs:critical-sequence}, there is a vertex $w$ in $\bar P\cap H\cap X$. Since $w\in X\cap H = X\cap (C\setminus \bar P[C])$, we have that $w\in C$ and $w\notin \bar P[C]$. 
    However, by the definition of $\bar P[C]$, $\bar P[C]$ must contain every vertex of $\bar P\cap C$, which includes $w$. This contradicts $w\notin \bar P[C]$.
    Therefore, $|X\cap (C\setminus \bar P[C])| < |X\cap H|$ for every critical bag $X$ of $H$.
    This completes the proof.
\end{proof}

The following observation is obvious from the definition of the critical path and critical sequence, but it is crucial for our analysis.
\begin{observation} \label{obs:path}
    Let $\bar P$ be a critical sequence of $H$ with respect to a critical path $P$. Then, for every strongly connected subgraph $C$ of $H$, $\bar P[C]$ is a subsequence of a shortest path in $G$. 
\end{observation}
\begin{proof}
    By the definition of a critical path, $P$ is a shortest path in $G$. Since $\bar P$ is a subsequence of the vertices in $P$ by the definition of the critical sequence, and because $\bar P[C]$ is a subsequence of $\bar P$, it follows that $\bar P[C]$ is also a subsequence of the shortest path $P$ in $G$. This completes the proof.
\end{proof}

\begin{figure}
    \centering
    \includegraphics[scale=0.75]{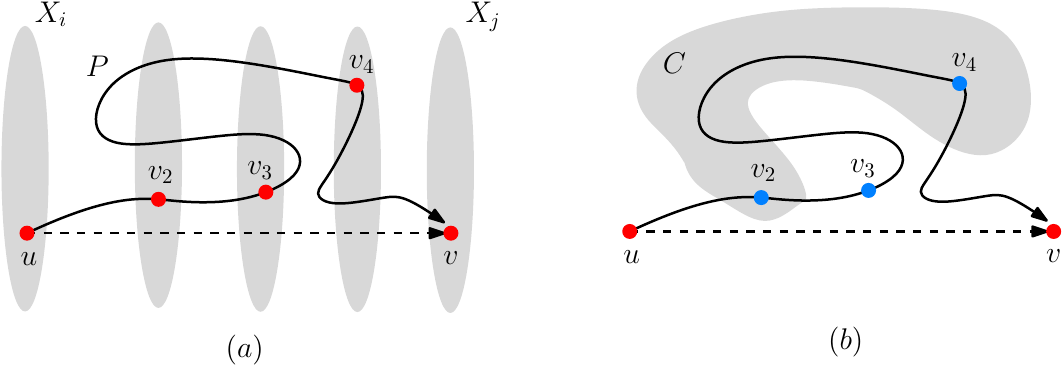}
    \caption{Illustration of the critical path and critical sequence of $H$. (a) The gray shaded region denotes the critical bags of $H$. A critical path $P$ connects a vertex $u$ in the first critical bag and a vertex $v$ in the last critical bag.  A critical sequence $\bar P$ of $H$ with respect to $P$ is $\langle u,v_2,v_3,v_4,v\rangle$. (b) For a subgraph $C$ of $H$, the surviving subsequence $\bar P[C] = \langle x=v_2,v_3,v_4=y\rangle$ is depicted as blue points.}
    \label{fig:critical-sequence}
\end{figure}

\paragraph{Algorithm.}
Then, we describe our algorithm $\rp(G, \mathcal P, H, \bar P, \Delta)$. See~\Cref{alg:rp} for the pseudocode.
It receives as input a directed graph $G$, a path partition $\mathcal P$ of $G$, a strongly connected subgraph $H$ of $G$, a critical sequence $\bar P$ of $H$ with respect to $P$, and a parameter $\Delta>0$.
Moreover, it outputs a random cutset $S$ of $E(H)$ such that for every SCC $C$ of $H\setminus S$, either $\ipw{C}<\ipw{H}$ or the weak diameter of $C$ is at most $12\Delta$.

The algorithm is a recursive procedure with respect to $H$ and $\bar P$.
Depending on the forward distance $f_{\bar P}(H)$ and the backward distance $b_{\bar P}(H)$ of $H$, the procedure proceeds in one of three stages.
The first stage is executed when $f_{\bar P}(H)=\tau >\Delta$.
At this stage, we sample a random radius $R$ uniformly at random from $[\frac{1}{2}\tau, \frac{2}{3}\tau]$. Then, we compute two balls, an out-ball $B_1=\outb{G}{u,R}$ and an in-ball $B_2=\inb{G}{v,R}$, compute a set $S_1$ of boundary edges of $B_1$ and $B_2$ that intersect $H$, i.e., $S_1=H\cap (\delta^+(B_1)\cup \delta^-(B_2))$.
Next, we recursively apply the algorithm to each SCC $C$ of $H\setminus S_1$ which is fully contained in $B_1\cup B_2$.
In addition, we pass the surviving subsequence $\bar P[C]$ as the critical sequence to the next recursive call.
This completes the description of the first stage.

The second stage is executed when $f_{\bar P}(H) \leq \Delta$ and $b_{\bar P}(H)=\tau' > 8\Delta$.
At this stage, we sample a random radius $R$ uniformly at random from $[\frac{1}{2}\tau', \frac{13}{24}\tau']$. Then, we compute two balls, an in-ball $B_1=\inb{G}{u,R}$ and an out-ball $B_2=\outb{G}{v,R}$, compute a set $S_1$ of boundary edges of $B_1$ and $B_2$ that intersect $H$, i.e., $S_1=H\cap (\delta^-(B_1)\cup \delta^+(B_2))$.
Next, we recursively apply the algorithm to each SCC of $H\setminus S_1$ which is fully contained in $B_1\cup B_2$.
In addition, we pass the surviving subsequence $\bar P[C]$ to the next recursive call.
This completes the description of the second stage.

The third stage is executed when $f_{\bar P}(H) \leq \Delta$ and $b_{\bar P}(H) \leq 8\Delta$.
We sample two random radii $R_1,R_2$ uniformly at random from $[\Delta, 2\Delta]$ and $[9\Delta, 10\Delta]$, respectively. Then, we compute two balls, an out-ball $B_1=\outb{G}{u,R_1}$ and an in-ball $B_2=\inb{G}{v,R_2}$, return $H\cap (\delta^+(B_1)\cup \delta^-(B_2))$. 

Finally, we provide a crucial technical remark:
As an extreme case, our algorithm immediately returns the empty set whenever $\bar P[C]$ is empty (lines 1-2). Intuitively, due to Lemma~\ref{lem:critical-path}, $\bar P[C]=\emptyset$ implies $\ipw{C}<\ipw{H}$, which satisfies our structural goal.
\begin{algorithm}[H] 
    \caption{$\rp(G, \mathcal P, H, \bar P=\langle u,\ldots, v\rangle, \Delta)$}
	\begin{algorithmic}[1] 
        \If{$\bar P=\emptyset$} 
        \State \Return $\emptyset$
        \EndIf
        \State $\tau, \tau' \gets d_G(u,v), d_G(v,u)$
        \State $S_1,S_2\gets \emptyset$
        \Statex
        \commentline{Stage 1. Reducing forward distance}
        \If{$\tau > \Delta$}
        \State Sample $R$ uniformly at random in the range $[\frac{1}{2}\tau, \frac{2}{3}\tau]$
        \State $B_1\gets \outb{G}{u,R}$
        \State $B_2\gets \inb{G}{v,R}$
        \State $S_1\gets$ $H\cap (\delta^+(B_1)\cup \delta^-(B_2))$
        \State $\mathcal C\gets$ SCCs of $H\setminus S_1$ that are fully contained in $B_1 \cup B_2$
        \For{$C\in \mathcal C$}
        \State $S_2\gets S_2\cup \rp(G, \mathcal P ,C, \bar P[C], \Delta)$
        \EndFor
        \State \Return $S_1\cup S_2$
        \EndIf
        \Statex
        \commentline{Stage 2. Reducing backward distance}
        \If{$\tau\leq \Delta$ and $\tau' > 8\Delta$}
        \State Sample $R$ uniformly at random in the range $[\frac{1}{2}\tau', \frac{13}{24}\tau']$
        \State $B_1\gets \inb{G}{u,R}$
        \State $B_2\gets \outb{G}{v, R}$
        \State $S_1\gets$ $H\cap (\delta^-(B_1)\cup \delta^+(B_2))$
        \State $\mathcal C\gets$ SCCs of $H\setminus S_1$ that are fully contained in $B_1\cup B_2$
        \For{$C\in \mathcal C$}
        \State $S_2\gets S_2\cup \rp(G, \mathcal P ,C, \bar P[C], \Delta)$
        \EndFor
        \State \Return $S_1\cup S_2$
        \Statex
        \commentline{Stage 3. Bounding weak diameter}
        \Else
            \State Sample $R_1$ uniformly at random from $[\Delta, 2\Delta]$, and $R_2$ uniformly at random from $[9\Delta, 10\Delta]$
            \State $B_1\gets \outb{G}{u,R_1}$
            \State $B_2\gets \inb{G}{u,R_2}$
            \State \Return $H\cap (\delta^+(B_1)\cup \delta^-(B_2))$
        
        \EndIf
        \end{algorithmic}
	\label{alg:rp}
\end{algorithm}

\subsection{Analysis} \label{sec:LDD-analysis}
\reducepathwidth*
In this section, we show that our algorithm $\rpp$ satisfies the condition of Lemma~\ref{lem:reduce-pathwidth}. As we discussed in the previous section, this concludes the proof of Lemma~\ref{lem:ldd-pathpartition} and therefore Theorem~\ref{thm:ldd}.
As a base case, the lemma is obvious when $\ipw{H}=0$, which implies $H=\emptyset$ and $S=\emptyset$.
Hence, let us assume $\ipw{H}>0$.

Before we prove Lemma~\ref{lem:reduce-pathwidth}, we prove two technical lemmas.
First, we verify the exceptional case (lines 1 and 2) such that if $\bar P$ is the empty set, then $\ipw{H}$ is reduced compared to that of the parent instance.
\begin{lemma}\label{lem:critical-empty}
    Suppose $\rp(G,\mathcal P, H', \bar P', \Delta)$ calls $\rp(G,\mathcal P, H, \bar P, \Delta)$ and $\bar P=\emptyset$. Then, $\ipw{H} < \ipw{H'} $.
\end{lemma}
\begin{proof}
    We show the statement in the case where the recursive call is executed in line 12. The other case where the recursive call is executed in line 21, is analogous.
    Recall that $\bar P=\bar P'[H]$ is a subsequence of $\bar P'$ that contains every vertex of $\bar P'\cap H$.
    Hence, $\bar P=\bar P'[H]=\emptyset$ implies $\bar P'\cap H=\emptyset$.
    Due to Lemma~\ref{lem:critical-path}, $\ipw{H} = \ipw{H\setminus \bar P'} < \ipw{H'}$. 
    This completes the proof.
\end{proof}

Next, we prove the technical lemma that settles the stages of the recursion tree.
\begin{lemma} \label{lem:gen-order}
    For every recursive call in which the procedure is in Stage 2, i.e., operates lines 15-22, all subsequent recursive calls remain in Stage 2 or enter Stage 3.  
\end{lemma}
\begin{proof}
    Let $\rp(G,\mathcal P,H, \bar P, \Delta)$ be a recursive call in Stage 2. For a strongly connected subgraph $C$ of $H$ that is fully contained in $B_1\cup B_2$, suppose it invokes a recursive call  $\rp(G,\mathcal P,C,\bar P[C], \Delta)$.
    Let $\bar P=\langle u,\ldots, v\rangle$, and let $\bar P[C]=\langle u',\ldots, v'\rangle$. Since $\bar P[C]$ is a consecutive subsequence of $\bar P$ and $\bar P$ is itself a subsequence generated from a shortest path in $G$ by Observation~\ref{obs:path}, we have that $d_G(u', v') \leq d_G(u, v)$.
    Since $\rp(G,\mathcal P,H,\bar P, \Delta)$ is in Stage 2, it must have satisfied the conditional statement in line 14, which implies $d_G(u,v) \leq \Delta$. 
    Therefore, $$f_{\bar P}(C) = d_G(u', v') \leq d_G(u, v) \leq \Delta.$$
    Hence, we conclude that the forward distance of the recursive call remains upper-bounded by $\Delta$.
    Consequently, the recursive call cannot satisfy the conditional statement in line 5 (Stage 1), meaning that it must either remain in Stage 2 or enter Stage 3.
    By repeating this argument, we conclude that all subsequent recursive calls remain in Stage 2 or enter Stage 3.
    This completes the proof.
\end{proof}

Now, we are ready to prove Lemma~\ref{lem:reduce-pathwidth}.
We prove the two items of Lemma~\ref{lem:reduce-pathwidth} separately, through~\Cref{lem:gen-correct,lem:gen-prob}, respectively. 
\begin{lemma} \label{lem:gen-correct}
    Let $S=\rp(G,\mathcal P,H, \bar P, \Delta)$ with a nonempty $\bar P$, and let $C$ be an SCC of $H\setminus S$. Then, either the weak diameter of $C$ is at most $12\Delta$ or $\ipw{C} <\ipw{H}$. 
\end{lemma}
\begin{proof}
    Let $\bar P=\langle u,\ldots, v\rangle$.
     We prove the lemma by induction on the height of the recursion tree generated by $\rpp$.
    For the base case, suppose the height of the recursion tree is 1. This implies that 
    the recursive call is in Stage 3.
    In this stage, the conditional statements (lines 5 and 14) ensure that 
    $d_G(u,v) \leq \Delta$ and $d_G(v,u)\leq 8\Delta$.
    Due to Lemma~\ref{lem:twoball-fin}, $C$ is either contained in $B_1\cap B_2$ or disjoint from $B_1\cap B_2$.
    Since $\bar P$ is a subsequence of a shortest path in $G$ by Observation~\ref{obs:path},
    for every $w\in \bar P$, $d_G(u,w)\leq d_G(u,v)\leq \Delta \leq R_1$. Hence, $\bar P\subset B_1$.
    In addition, for every $w\in \bar P$, the triangle inequality yields
    $$d_G(w, u) \leq d_G(w, v) + d_G(v,u) \leq d_G(u,v)+d_G(v,u) \leq \Delta+8\Delta\leq R_2,$$
    where the second inequality follows because $w$ appears before $v$ along the shortest path.
    Then, it follows that $\bar P\subset B_2$ and therefore $\bar P\subset B_1\cap B_2$.
    Consequently, if $C$ is disjoint from $B_1\cap B_2$, then $C$ is disjoint from $\bar P$.
    This disjointness implies $\ipw{C}<\ipw{H}$ due to Lemma~\ref{lem:critical-path}.

    Next, we analyze the weak diameter of an SCC contained in $B_1\cap B_2$.
    For every $w,w'\in B_1\cap B_2$, 
$$d_G(w,w')\leq d_G(w,u)+d_G(u,w')\leq R_2+R_1 \leq 12\Delta.$$
This shows that the weak diameter of an SCC contained in $B_1\cap B_2$ is at most $12\Delta$.
Thus, we conclude that either the weak diameter of $C$ is at most $12\Delta$, or $\ipw{C}<\ipw{H}$. This completes the proof for the base case.

    For the inductive argument, suppose $\rpp(G,\mathcal P,H,\bar P, \Delta)$ generates a recursion tree of depth $i>1$. 
    Then, $\rpp$ is either in Stage 1 or in Stage 2.
    We prove the inductive argument for the case where it is in Stage 1. 
    Due to Lemma~\ref{lem:twoball-fin}, every SCC of $H\setminus S_1$ (after line 9) is either fully contained in $B_1$, contained in $B_2\setminus B_1$, or disjoint from $B_1\cup B_2$.
    Since $S=S_1\cup S_2$ is a superset of $S_1$, $C$ is either fully contained in $B_1$, contained in $B_2\setminus B_1$, or disjoint from $B_1\cup B_2$.
    
    For the first case, suppose $C$ is contained in $B_1$. 
    Since $H\setminus S$ is a subgraph of $H\setminus S_1$, there is an SCC $C'$ of $H\setminus S_1$ that contains $C$.
    If $\bar P[C']=\emptyset$, then $\ipw{C'} < \ipw{H}$ by Lemma~\ref{lem:critical-empty}. 
    Since $C\subseteq C'$, we have that $\ipw{C} <\ipw{H}$, as desired.
    Otherwise, let us assume $\bar P[C']$ is nonempty.
    Then, due to the induction hypothesis on $\rpp(G, \mathcal P ,C', \bar P[C'], \Delta)$, after line 12, the weak diameter of every SCC of $C'\setminus S_2$ is at most $12\Delta$.
    Since $C\subseteq C'\setminus S$, $C$ is a subgraph of $C'\setminus S_2$. Therefore, the weak diameter of $C$ is at most $12\Delta$.
    Overall, we have that either $\ipw{C}<\ipw{H}$ or the weak diameter of $C$ is at most $12\Delta$.
    This completes the inductive proof for the first case where $C$ is contained in $B_1$.
    
    Similarly, in the second case where $C$ is fully contained in $B_2\setminus B_1$, we can analogously verify that either $\ipw{C} < \ipw{H}$ holds by Lemma~\ref{lem:critical-empty}, or the induction hypothesis applies, which implies that the weak diameter of $C$ is at most $12\Delta$.
    
    For the last case, suppose $C$ is disjoint from $B_1\cup B_2$.
    For any vertex $w\in \bar P$, since $d_G(u,v)=\tau$ and $w$ lies on a shortest $u$-$v$ path by Observation~\ref{obs:path}, we have that either $d_G(u,w)\leq \frac{1}{2}\tau$ or $d_G(w,v)\leq \frac{1}{2}\tau$.
    Since $\frac{1}{2}\tau \leq R$ by construction, either $w\in B_1$ or $w\in B_2$ must hold.
    This implies that $\bar P\subset B_1\cup B_2$ and thus $\bar P\cap C = \emptyset$.
    By the definition of $\bar P[C]$, we have that $\bar P[C]$ is empty.
    Then, due to Lemma~\ref{lem:critical-empty}, $\ipw{C} < \ipw{H}$.
    This completes the inductive proof for the case where $\rpp$ is in Stage 1.

    \medskip
    Next, we prove the inductive argument for the case when $\rpp$ is in Stage 2.
    The proof of this case mirrors the proof of the previous case. Due to Lemma~\ref{lem:twoball-fin}, every SCC of $H\setminus S_1$ (after line 18) is either fully contained in $B_1$, contained in $B_2\setminus B_1$, or disjoint from $B_1\cup B_2$.
    The first two cases that $C\subseteq B_1$ or $C\subseteq B_2\setminus B_1$ can be handled by exactly the same argument as in the previous case. On the other hand, we use a different argument for the last case.
    
    Suppose $C$ is in the last case, i.e., $C$ is disjoint from $B_1\cup B_2$.
    In Stage 2, unlike Stage 1, $\bar P$ is not guaranteed to be contained in $B_1\cup B_2$ since for a middle point $x\in \bar P$, both $d_G(x,u)$ and $d_G(v,x)$ might be much larger than $d_G(v,u)$. This prevents us from directly applying Lemma~\ref{lem:critical-empty}.
    To circumvent this obstacle, we establish an alternative argument as follows. 
    Let $P'$ be a shortest $v$-$u$ path of length $\tau'$ in $G$.
    Since $R\geq \frac{1}{2}\tau'$, we have that $P'\subseteq B_1\cup B_2$. 
    We then claim that this reversed shortest path $P'$ can play the role of a ``critical path". See Figure~\ref{fig:overview-hbc}(b) for an illustration. 
    \begin{claim}
        $\ipw{C\setminus P'} < \ipw{H}$.
    \end{claim}
    \begin{proof}[Proof of the claim.]
        Let $X_u$ and $X_v$ be the bags in the path partition $(\mathcal P, \mathcal X)$ that contain $u$ and $v$, respectively. 
        Due to Lemma~\ref{lem:critical-path}, $\ipw{H\setminus \bar P} < \ipw{H}$. 
        This upper bound implies that for every critical bag $X$ of $H$, $|X\cap H\cap \bar P|$ is nonempty.
        Moreover, by the definition of the critical sequence, every critical bag $X$ of $H$ must lie between $X_u$ and $X_v$ in the path partition.

        Due to the definition of the path partition, $P'$ must intersect every bag between $X_u$ and $X_v$ in the path partition.
        This implies that $P'$ intersects every critical bag of $H$.
        Let $X$ be a critical bag of $H$.
        If $C$ retains all vertices of this bag, i.e., $|X\cap C| = |X\cap H|$, then $|X\cap C\cap P'|=|X\cap H\cap P'|$ is nonempty.
        For any such critical bag $X$, we have that:
        $$|X\cap (C\setminus P')| = |X\cap C| - |X\cap (C\cap P')| = |X\cap H| - |X\cap C\cap P'| < |X\cap H|.$$
        On the other hand, if $|X\cap C| < |X\cap H|$ initially, we trivially have that $|X\cap (C\setminus P')| < |X\cap H|$.
        Thus, we conclude that for every critical bag $X$ of $H$, $|X\cap (C\setminus P')| < |X\cap H|$.
        This establishes that $\ipw{C\setminus P'} < \ipw{H}$.
    \end{proof}
    Since $C$ is disjoint from $B_1\cup B_2$, we have $C\cap P'=\emptyset$.
    Due to the claim, we have that $$\ipw{C}=\ipw{C\setminus P'} < \ipw{H},$$ as desired.
    This completes the inductive proof for the case where $\rpp$ is in Stage 2.

    In all cases, we conclude that either the weak diameter of $C$ is at most $12\Delta$ or $\ipw{C}< \ipw{H}$. This completes the proof.
\end{proof}

Next, we bound the edge-cut probability.
\begin{lemma} \label{lem:gen-prob} 
     Let $S=\rp(G,\mathcal P,H,\bar P, \Delta)$. For every $e\in E(H)$, 
     $$\Pr[e\in S]\leq 164\cdot \frac{w(e)}{\Delta}.$$
\end{lemma}
\begin{proof}
    If $\bar P=\emptyset$, $\Pr[e\in S]=0$ (line 2). Hence, let us assume that $\bar P\neq \emptyset$.
    Let $e=(x,y)$, $\bar P=\langle u=v_1,v_2,\ldots, v_\ell=v\rangle$, $\tau=d_G(u,v)$ and $\tau'=d_G(v,u)$.
    We prove the lemma by case analysis on the stages.
    For the first case,  suppose $\rpp$ is a recursive call in Stage 3. 
    In this case, $e\in S$ if and only if $e\in \delta^+(B_1)$ or $e\in \delta^-(B_2)$.
    We analyze $\Pr[e\in \delta^+(B_1)]$. 
    By the construction of $B_1$, $e\in \delta^+(B_1)$ if and only if $d_G(u, x) \leq R_1$ and $d_G(u,y) >R_1$.
    Therefore, $\Pr[e\in \delta^+(B_1)] \leq \Pr[ d_G(u,x)\leq R_1< d_G(u,y)]$. Since we sample $R_1$ uniformly at random from $[\Delta, 2\Delta]$,
    we have that
    $$\Pr[e\in \delta^+(B_1)] \leq \Pr[ d_G(u,x)\leq R_1 < d_G(u,y)] \leq \frac{d_G(u,y)-d_G(u,x)}{2\Delta-\Delta}.$$
    By the triangle inequality, $d_G(u,y)-d_G(u,x) \leq d_G(x,y)$. Hence, we have that
    $$\Pr[e\in \delta^+(B_1)] \leq \frac{d_G(x,y)}{\Delta} \leq \frac{w(e)}{\Delta}.$$
    Symmetrically, since we sample $R_2$ uniformly at random from $[9\Delta, 10\Delta]$, we can easily show that $\Pr[e\in \delta^-(B_2)]  \leq \frac{w(e)}{\Delta}$.
    Due to the union bound of probability, we have that:
    $$\Pr[e\in S] \leq \Pr[e\in \delta^+(B_1)]+\Pr[e\in \delta^-(B_2)] \leq 2\cdot \frac{w(e)}{\Delta} < 164\cdot \frac{w(e)}{\Delta}.$$
    This completes the proof for the case where $\rpp$ is a recursive call in Stage 3.

    For the second case, suppose $\rpp$ is a recursive call in Stage 2.
    We claim the stronger, parameterized bound on $\Pr[e\in S]$ as follows.
    \begin{claim} \label{claim:case2}
        Suppose $\rp(G,\mathcal P,H,\bar P, \Delta)$ lies in Stage 2. Then,
        $$\Pr[e\in S] \leq \left(20- \frac{96\Delta}{\max( \tau', \frac{16\Delta}{3})}\right)\cdot \frac{w(e)}{\Delta}.$$
    \end{claim}
    \begin{proof}[Proof of the claim.]
        Recall that $\bar P=\langle u,\ldots, v\rangle$.
        We prove the claim by induction on the height of the recursion tree generated by $\rpp$.
        Note that the recursion depth is at least two since otherwise, the recursive call would lie in Stage 3.
        For the base case, suppose the recursion depth is 2. 
        This means that every further recursive call either lies in Stage 3, or terminates immediately. 
        From the above arguments about Stage 3,  $\Pr[e\in S_2] \leq 2\cdot \frac{w(e)}{\Delta}$ holds for both cases,.
        Since each edge participates in at most one SCC of $H\setminus S_1$, the same probability holds for the union of $S_2$ over all recursive calls.
        Since we compute two balls of random radius $R\in [\frac{1}{2}\tau', \frac{13}{24}\tau']$, using the same argument as in the base case (Stage 3), we can easily show that 
        $$\Pr[e\in S_1] \leq \Pr[e\in \delta^-(B_1)] + \Pr[e\in \delta^+(B_2)] \leq 2\cdot \left(\frac{w(e)}{\frac{13}{24}\tau' - \frac{1}{2}\tau'}\right) = 48 \cdot \frac{w(e)}{\tau'}.$$
        Also, since $\rpp(G,\mathcal P,H,\bar P, \Delta)$ lies in Stage 2, we have $\tau'>8\Delta$.
        Then, $144\cdot \frac{w(e)}{\tau'} < 18 \cdot \frac{w(e)}{\Delta}$.
        By subtracting $96\cdot \frac{w(e)}{\tau'}-2\cdot \frac{w(e)}{\Delta}$ for both sides, we have that:
        \begin{align} \label{eq:1-1}
            2\cdot \frac{w(e)}{\Delta} + 48\cdot \frac{w(e)}{\tau'} < 20\cdot \frac{w(e)}{\Delta} - 96\cdot \frac{w(e)}{\tau'}
        \end{align}
        Therefore, due to the union bound of probability, we have that:
        \begin{align*}
            \Pr[e\in S] \leq \Pr[e\in S_1] + \Pr[e\in S_2] &\leq 2\cdot \frac{w(e)}{\Delta} + 48\cdot \frac{w(e)}{\tau'}
            \\&\stackrel{\eqref{eq:1-1}}< 20\cdot \frac{w(e)}{\Delta} - 96\cdot \frac{w(e)}{\tau'}
            \\&\leq \left(20- \frac{96\Delta}{\max(\tau', \frac{16\Delta}{3})}\right)\cdot \frac{w(e)}{\Delta},
        \end{align*}
        where the last inequality holds since $\max(\tau', \frac{16\Delta}{3}) = \tau'$ follows directly from the fact that $\tau' > 8\Delta > \frac{16\Delta}{3}$.
        This completes the proof of the base case.

         For the inductive case, note that we sample $R$ uniformly at random from $[\frac{1}{2}\tau', \frac{13}{24}\tau']$. Using the same argument as in the base case, we can derive that 
         $$\Pr[e\in S_1] \leq 48\cdot \frac{w(e)}{\tau'}.$$        
    Next, we analyze $\Pr[e\in S_2]$.
    Note that $\Pr[e\in S_2]$ is affected by at most one SCC of $H\setminus S_1$, say $C$, that contains $e$. Then, either $C$ is fully contained in $B_1$ or fully contained in $B_2\setminus B_1$ by Lemma~\ref{lem:twoball-fin}.
    For the first case, suppose $C$ is contained in $B_1$.
    If $\bar P[C]=\emptyset$, then the recursive call $\rp(G, \mathcal P, C, \bar P[C], \Delta)$ immediately returns the empty cutset, which implies $\Pr[e\in S_2]=0$. Hence, let us assume $\bar P[C]\neq \emptyset$.
    Let $x,y$ be the first and the last vertices of $\bar P[C]$, respectively. Since $x,y\in \bar P[C]\subseteq \bar P$ and $\bar P$ is a subsequence of a $u$-$v$ shortest path in $G$ by Observation~\ref{obs:path}, we have
    $d_G(u,x)\leq d_G(u,v)=\tau \leq \Delta$. In addition, we have $d_G(y,u) \leq R$ since $y\in C\subseteq B_1$. Combining these inequalities with the triangle inequality, we have that:
    \begin{align} \label{eq:2-2}
        b_{\bar P}(C) = d_G(y,x)\leq d_G(y,u) + d_G(u, x) \leq d_G(y,u) + d_G(u,v) \leq R + \Delta \leq \frac{13}{24}\tau' + \frac{1}{8}\tau' = \frac{2}{3}\tau'.
    \end{align}
    Here, the last inequality holds because the conditional statements for triggering Stage 2 guarantee $8\Delta < \tau'$ and $d_G(u,v) = \tau \leq \Delta$.
    Furthermore, due to Lemma~\ref{lem:gen-order}, $\rp(G, \mathcal P, C, \bar P[C], \Delta)$ is either in Stage 2, or in Stage 3.
    Due to the necessary condition for Stage 2, $\tau' > 8\Delta$. Hence, we have that 
    \begin{align} \label{eq:2-1}
        \frac{2}{3}\tau'>\frac{16\Delta}{3}.
    \end{align}
    If it is in Stage 2, due to the induction hypothesis, we have that 
    $$\Pr[e\in S_2] \leq \left(20- \frac{96\Delta}{\max(d_G(y,x), \frac{16\Delta}{3})}\right)\cdot \frac{w(e)}{\Delta} \stackrel{\eqref{eq:2-2}}\leq \left(20- \frac{96\Delta}{\max(\frac{2}{3}\tau', \frac{16\Delta}{3})}\right)\cdot \frac{w(e)}{\Delta} \stackrel{\eqref{eq:2-1}}\leq \left(20-\frac{144\Delta}{\tau'}\right)\cdot \frac{w(e)}{\Delta}.$$
    Otherwise, if it is in Stage 3, from the arguments about Stage 3, 
    $$\Pr[e\in S_2] \leq 2\cdot \frac{w(e)}{\Delta} =(20-18)\cdot \frac{w(e)}{\Delta} = \left(20-\frac{96\Delta}{\frac{16\Delta}{3}}\right)\cdot \frac{w(e)}{\Delta} \stackrel{\eqref{eq:2-1}}\leq \left(20-\frac{144\Delta}{\tau'}\right)\cdot \frac{w(e)}{\Delta}.$$
    Thus, for all cases, we have that $\Pr[e\in S_2] \leq (20-\frac{144\Delta}{\tau'})\cdot \frac{w(e)}{\Delta}$.

    The case where $C \subseteq B_2 \setminus B_1$ is analogous. 
    In this case, let $x',y'$ be the first and the last vertices of $\bar P[C]$, respectively.
    Since $x'\in B_2$, we have $d_G(v, x') \leq R\leq \frac{13}{24}\tau'$.
    Due to the triangle inequality and the fact that $\tau'>8\Delta$ we have that:
    $$d_G(y',x') \leq d_G(y',v) + d_G(v, x') \leq d_G(u,v) + d_G(v,x')\leq \Delta + R \leq \frac{1}{8}\tau' + \frac{13}{24}\tau'=\frac{2}{3}\tau'.$$
    Then, repeating the same argument as above yields $\Pr[e\in S_2] \leq (20-\frac{144\Delta}{\tau'})\cdot \frac{w(e)}{\Delta}$.
    Finally, due to the union bound of probability, we have that:
   \begin{align*}
        \Pr[e\in S] &\leq \Pr[e\in S_1] + \Pr[e\in S_2] 
        \\&\leq 48\cdot \frac{w(e)}{\tau'} + \left(20-\frac{144\Delta}{\tau'}\right)\cdot \frac{w(e)}{\Delta}
        \\&= \left(20-\frac{96\Delta}{\tau'}\right)\cdot \frac{w(e)}{\Delta} \stackrel{\eqref{eq:2-1}}= \left(20- \frac{96\Delta}{\max(\tau', \frac{16\Delta}{3})}\right)\cdot \frac{w(e)}{\Delta}.
    \end{align*}
    This completes the inductive proof.
    \end{proof}

    For the last case, suppose $\rpp$ is a recursive call that lies in Stage 1.
    Using Claim~\ref{claim:case2}, we show the stronger probability bound as follows. 
    \begin{claim} \label{claim:case1}
        Suppose $\rp(G,\mathcal P,H,\bar P, \Delta)$ lies in Stage 1. Then,
        $$\Pr[e\in S] \leq \left(164- \frac{96\Delta}{\max(\tau, \frac{2\Delta}{3})}\right)\cdot \frac{w(e)}{\Delta}.$$
    \end{claim}
    \begin{proof}[Proof of the claim.]
        Recall that $\bar P=\langle u,\ldots, v\rangle$. The logical flow of the proof mirrors that of Claim~\ref{claim:case2}.
        We prove the claim by induction on the height of the recursion tree generated by $\rpp$.
        Note that the recursion depth is at least two since otherwise, the recursive call would lie in Stage 3. 
        For the base case, suppose the recursion depth is two. 
        Note that any subsequent recursive call cannot lie in Stage 2 because that would imply that the recursion depth is greater than two. Hence, every further recursive call either lies in Stage 3 or terminates immediately. Using the same argument as in Claim~\ref{claim:case2}, we can check that $$\Pr[e\in S_1] \leq 48\cdot \frac{w(e)}{\tau},\text{ and } \Pr[e\in S_2]\leq 2\cdot \frac{w(e)}{\Delta}.$$
        Due to the necessary condition for Stage 1, we have $\tau > \Delta$. Then, $144\cdot \frac{w(e)}{\tau} < 144\cdot \frac{w(e)}{\Delta}$. By subtracting $96\cdot \frac{w(e)}{\tau}-2\cdot \frac{w(e)}{\Delta}$ for both sides, we obtain:
        \begin{align} \label{eq:1-3}
            2\cdot \frac{w(e)}{\Delta} + 48\cdot \frac{w(e)}{\tau} \leq 146\cdot \frac{w(e)}{\Delta} - 96\cdot \frac{w(e)}{\tau}
        \end{align}
        Then, due to the union bound of probability, 
        \begin{align*}
            \Pr[e\in S] \leq \Pr[e\in S_1]+ \Pr[e\in S_2] &\leq 2\cdot \frac{w(e)}{\Delta} + 48\cdot \frac{w(e)}{\tau}
            \\&\stackrel{\eqref{eq:1-3}}\leq 146\cdot \frac{w(e)}{\Delta} - 96\cdot \frac{w(e)}{\tau}
            \\&< \left(164- \frac{96\Delta}{\max(\tau, \frac{2\Delta}{3})}\right)\cdot \frac{w(e)}{\Delta}.
        \end{align*}
        This completes the proof for the base case.
                
        For the inductive case, using the same argument as in Claim~\ref{claim:case2}, we can check that 
        $$\Pr[e\in S_1] \leq \frac{48}{\tau}\cdot w(e).$$
         Next, we analyze $\Pr[e\in S_2]$. As in Claim~\ref{claim:case2}, this probability is affected by at most one SCC of $H\setminus S_1$, say $C$, that contains $e$.
         Then, $C$ is either fully contained in $B_1$ or fully contained in $B_2\setminus B_1$ by Lemma~\ref{lem:twoball-fin}. For the first case, suppose $C$ is contained in $B_1$. Let $x,y$ be the first and the last vertices of $\bar P[C]$, respectively. Since $\bar P[C]$ is a subsequence of a shortest $u$-$v$ path in $G$ by Observation~\ref{obs:path} and $y\in B_1$, we have $d_G(x,y) \leq d_G(u,y) \leq R \leq \frac{2}{3}\tau$. 
         The difference from Claim~\ref{claim:case2} is that $\rp(G,\mathcal P, C, \bar P[C], \Delta)$ can now belong to Stage 1, Stage 2, or Stage 3, or it may satisfy $\bar P[C]=\emptyset$.

         \medskip
         If $\bar P[C]=\emptyset$, then $\Pr[e\in S_2]=0$. If it belongs to Stage 3, from the arguments for Stage 3, $\Pr[e\in S_2] \leq 2\cdot \frac{w(e)}{\Delta} < (164-\frac{144\Delta}{\tau})\cdot \frac{w(e)}{\Delta}$.
         If it belongs to Stage 2, due to Claim~\ref{claim:case2}, we have that:
         \begin{align*}
             \Pr[e\in S_2]&\leq \left(20 - \frac{96\Delta}{\max(d_G(y, x), \frac{16\Delta}{3})}\right)\cdot \frac{w(e)}{\Delta}
             \\&\leq 20\cdot \frac{w(e)}{\Delta} = (164 - 144)\cdot \frac{w(e)}{\Delta}
             \\&< \left(164 - \frac{144\Delta}{\tau}\right)\cdot \frac{w(e)}{\Delta},
         \end{align*}
         where the last inequality follows from the fact that $\tau > \Delta$.
         Finally, if it belongs to Stage 1, due to the induction hypothesis, 
          \begin{align*}
        \Pr[e\in S_2] &\leq \left(164 - \frac{96\Delta}{\max(d_G(x, y), \frac{2\Delta}{3})}\right)\cdot \frac{w(e)}{\Delta}
        \\&\leq \left(164 - \frac{96\Delta}{\max(\frac{2}{3}\tau, \frac{2\Delta}{3})}\right)\cdot \frac{w(e)}{\Delta}
        \\&= \left(164-\frac{144\Delta}{\tau}\right) \cdot \frac{w(e)}{\Delta}.
    \end{align*}
    Here, the second inequality follows from the fact that $\tau>\Delta$.
    Therefore, in all cases, $$\Pr[e\in S_2]\leq \left(164-\frac{144\Delta}{\tau}\right)\cdot \frac{w(e)}{\Delta}.$$
    
    The case where $C \subseteq B_2 \setminus B_1$ is analogous. 
    In this case, let $x',y'$ be the first and the last vertices of $\bar P[C]$, respectively. Then,
    observe that $d_G(x', y') \leq d_G(x', v) \leq R\leq \frac{2}{3}\tau$.
    Repeating the same argument as above yields $\Pr[e\in S_2] \leq (164-\frac{144\Delta}{\tau})\cdot \frac{w(e)}{\Delta}$.

    Finally, due to the union bound of probability we have that:
    \begin{align*}
          \Pr[e\in S] &\leq \Pr[e\in S_1] + \Pr[e\in S_2] 
        \\&\leq 48\cdot \frac{w(e)}{\tau} + \left(164-\frac{144\Delta}{\tau}\right)\cdot \frac{w(e)}{\Delta}
        \\&= \left(164-\frac{96\Delta}{\tau}\right)\cdot \frac{w(e)}{\Delta} = \left(164- \frac{96\Delta}{\max(\tau, \frac{2\Delta}{3})}\right)\cdot \frac{w(e)}{\Delta}.
    \end{align*}
    This completes the inductive proof.
    \end{proof}
    Then, Lemma~\ref{lem:gen-prob} follows from the analysis of Stage 3, Claim~\ref{claim:case2}, and Claim~\ref{claim:case1}. This completes the proof.
\end{proof}

\section{Directed Sparsest Cut Problem for Bounded-Treewidth Digraphs} 
\label{sec:quasipartition}
In this section, we study the directed non-bipartite sparsest cut problem and its LP relaxation, focusing on the integrality gap.

\paragraph{Directed non-bipartite sparsest cut problem.}
Let $G$ be an edge-weighted directed graph with \emph{capacities} $c(e)$ for all edges $e\in E(G)$. Let $T$ be a set of terminal pairs $T=\{(s_1,t_1),(s_2,t_2),\ldots, (s_\ell,t_\ell)\}$ of vertices. For each $i\in [\ell]$, let $\textsf{dem}(i)$ be a non-negative real number, called the \emph{demand}, for the terminal pair $(s_i, t_i)$. A \emph{cut} $S$ of $G$ is a subset $S\subseteq E(G)$ of edges.
For a cut $S$, let $c(S):=\sum_{e\in S} c(e)$ be the \emph{capacity} of $S$.
Let $I_S\subseteq [\ell]$ be the set of integers $i$ such that every $s_i$-$t_i$ path (not necessarily a shortest path) contains at least one edge of $S$.
We define the \emph{demand} of $S$ by $\textsf{dem}(S):=\sum_{i\in I_S}\dem(i)$.
The \emph{sparsity} of $S$ is defined by $c(S)/\dem(S)$.
The goal of the directed non-bipartite sparsest cut problem is to find a cut $S$ that minimizes the sparsity.

The integrality gap of the LP relaxation for this problem is called the directed multi-commodity flow cut-gap. For non-uniform demands, 
Hajiaghayi and R\"acke showed an $O(\sqrt n)$ upper bound for the flow-cut gap~\cite{hajiaghayi2006n}. This has been improved to $\tilde O(n^{11/23})$ by Agarwal, Alon, and Charikar~\cite{agarwal2007improved}.
On the other hand, the current best-known lower bound is due to Chuzhoy and Khanna~\cite{chuzhoy2009polynomial}, who showed a $\tilde \Omega(n^{1/7})$ lower bound for the flow-cut gap.

\medskip
M\'emoli, Sidiropoulos, and Sridhar~\cite{memoli2018quasimetric} showed that every $n$-vertex directed graph with treewidth $\tw$ admits an $(O(\tw \log n), \Delta)$-quasipartition for every parameter $\Delta>0$.
They used their quasipartition construction to show that the flow-cut gap is $O(\tw \log^2 n)$ for such graphs.
We now present the main result of this section.
\flowcutgap*

This improves upon the $\log n$ factor of the previously best-known bound from~\cite{memoli2018quasimetric}.
At a high level, we follow the procedure of~\cite{memoli2018quasimetric}, but provide a better analysis of their approach.

We briefly explain the approach of~\cite{memoli2018quasimetric}, which proved that the integrality gap is $O(\tw \log^2 n)$.
First, they showed that for every $\Delta>0$, $G$ admits a $(\beta, \Delta)$-quasipartition with $\beta=O(\tw \log n)$.
Using this quasipartition, they proved that the quasimetric space with respect to the shortest path metric of $G$ can be embedded into a convex combination of $0$-$1$ quasimetric spaces (see Definition~\ref{def:zeroone} and Lemma~\ref{lem:quasi-to-zeroone}) with a distortion of $O(\beta \log n)=O(\tw \log^2 n)$.
Finally, by using the known relationship between the convex combination of $0$-$1$ quasimetric spaces and the integrality gap for the directed non-bipartite sparsest cut problem~\cite{charikar2006directed}, they showed that the integrality gap is $O(\tw \log^2 n)$.

\begin{definition}
    A quasimetric space $M=(X,d)$ is a function $d$ over $X\times X$ that satisfies all but the symmetry axiom for a metric space.
In other words, $d$ satisfies the following properties:
\begin{itemize}
    \item $d(u,u)=0$ for all $u\in X$.
    \item (Triangle inequality) $d(u,v)\leq d(u,w)+d(w,v)$ for all $u,v,w\in X$.
\end{itemize}
\end{definition}
\begin{definition} \label{def:zeroone}
    A quasimetric space $M=(X,d)$ is a $0$-$1$ quasimetric space if every metric distance $d(u,v)$ is either 0 or 1.
\end{definition}
\begin{lemma}[Lemma~16 of~\cite{memoli2018quasimetric}] \label{lem:quasi-to-zeroone}
    Suppose a directed graph $G$ admits a $(\beta, \Delta)$-quasipartition for every $\Delta>0$. Then there is an embedding of $G$ into a convex combination of $0$-$1$ quasimetric spaces with a distortion of $O(\beta \log n)$.
\end{lemma}

\subsection{Deep dive into the approach of~\cite{memoli2018quasimetric}} \label{sec:quasipartition-tw}
In this section, we carefully review the previous approach of~\cite{memoli2018quasimetric} step-by-step, and then describe our improvement.
First, we analyze the $(O(\tw\log n), \Delta))$-quasipartition construction of~\cite{memoli2018quasimetric}.

The construction is achieved by a recursive approach. Given a tree decomposition $(T, \{B_t\}_{t\in T})$, we choose a central bag $B_t$, meaning that each subgraph of $G$ induced by the maximal subtree of $T\setminus \{t\}$ contains at most $\frac{2}{3} |V(G)|$ vertices. We then sample a radius $R$ uniformly at random from $[\frac{1}{2}\Delta, \Delta]$.
For each vertex $v\in B_t$, we compute two balls, an out-ball $B_1:=\outb{G}{v, R}$ and an in-ball $B_2:=\inb{G}{v,R}$, add $\delta^+(B_1) \cup \delta^-(B_2)$ to the cutset $S$, and recursively apply this procedure to each subgraph induced by the connected components of $T\setminus \{t\}$.
Note that any $x$-$y$ path intersecting $v$ such that $d_G(x,y)>2\Delta$ must intersect $\delta^+(B_1) \cup \delta^-(B_2)$ because, due to the triangle inequality, we have that:
$$\max(d_G(x,v), d_G(v,y)) \geq \frac{1}{2}(d_G(x,v) + d_G(v,y)) \geq \frac{1}{2}d_G(x,y) > \Delta.$$
This implies that for every path (not necessarily a shortest path) connecting two vertices from different trees of $T\setminus \{t\}$, either it intersects $S$, or it satisfies $d_G(x,y)\leq 2\Delta$.
This justifies the recursion: after the recursion terminates, we can conclude that for every $x$-$y$ path in $G\setminus S$, we have $d_G(x,y) \leq 2\Delta$.

We briefly analyze the probability that a single edge belongs to the cutset, i.e., $\Pr[e\in S]$. Since $B_t$ plays the role of a separator, each edge $e$ is affected by at most one instance at each level of the recursion. Moreover, since $B_t$ is a central bag, the recursion depth is $O(\log n)$.
At a single level, $e$ is affected by at most $2|B_t|=2(\tw+1)$ balls. Hence, due to the union bound on the probability, 
\begin{align*}
    \Pr[e\in S]\leq O(\log n) \cdot O(\tw) \cdot  \frac{w(e)}{\Delta} = O(\tw \log n) \cdot \frac{w(e)}{\Delta}.
\end{align*}
This completes the analysis of the construction of the $(O(\tw \log n), \Delta)$-quasipartition.

\medskip
Next, we analyze the second part which gives an embedding of a quasimetric space into a convex combination of $0$-$1$ quasimetric spaces (Lemma~\ref{lem:quasi-to-zeroone}).
For simplicity, suppose every edge weight of $G$ is a positive integer and $\max_{u,v}\{d_G(u,v)\}=2^k$ for some integer $k=O(\log n)$.
Then, for each scale $i\in \{0,1,\ldots, k\}$, we compute a $(\beta, 2^i)$-quasipartition $\mathcal D_i$ of $G$ using the above procedure.
Importantly, each quasipartition $\mathcal D_i$ is a distribution over a cutset $S_i\subset E(G)$.
Using these, we compute a random quasipartition of $G$ according to the following procedure. Let $c=2^{k+1}-1$. 
First, we select one $\mathcal D_i$ among $\{\mathcal D_0, \mathcal D_1, \ldots, \mathcal D_k\}$ with probability $\frac{2^i}{c}$.
Then, we pick a random cutset $S$ from the chosen $\mathcal D_i$.
This yields a distribution $\mathcal D$ over the cutsets. 

Next, for each (random) cutset $S$ sampled from $\mathcal D$, 
we define a $0$-$1$ quasimetric space $(V(G), d_S)$ where $d_S(x,y)=0$ if there is an $x$-$y$ path in $G\setminus S$ and $d_S(x,y)=1$ otherwise. 
This yields a distribution over $0$-$1$ quasimetric spaces which can be expressed as a convex combination of $0$-$1$ quasimetric spaces.
Let $(V(G), d^*)$ be the resulting quasimetric space.
Then, we have 
\begin{align}
    d^*(u,v) 
    &= \Pr_{S\sim \mathcal D}[\text{there is no } u\text{-}v \text{ path in } G\setminus S] \label{eq:5-16} \\
    &= \sum_{i=0}^k \frac{2^i}{c} \Pr_{S_i\sim \mathcal D_i}[\text{there is no } u\text{-}v \text{ path in } G\setminus S_i] \\
    &\leq \sum_{i=0}^k \frac{2^i}{c} \beta \cdot\frac{d_G(u,v)}{2^i}
    \\
    &= \frac{(k+1)\cdot \beta}{c} d_G(u,v). \label{eq:5-20}
\end{align}
For the other direction,~\cite{memoli2018quasimetric} showed that $d^*(u,v)\geq \frac{1}{c} \cdot d_G(u,v)$, yielding a distortion of $(k+1)\beta = O(\beta \log n)$.

\subsection{Removing the extra \texorpdfstring{$O(\log n)$} factor from Lemma~\ref{lem:quasi-to-zeroone}}
In this section, we provide a better analysis that improves the distortion by a factor of $O(\log n)$.
We focus on the $O(\log n)$ additional factor arising from Lemma~\ref{lem:quasi-to-zeroone}. This factor comes from combining the distortion contributions across $O(\log n)$ scales, where each scale contributes independently. 

Our idea is to avoid this factor by considering all scales simultaneously.
To formalize this idea, we revisit the construction in Section~\ref{sec:quasipartition-tw}.
For a fixed edge $e=(x,y)$ and a scale $i$, 
we compute a set $V(e)\subset V(G)$ of size $O(\tw \log n)$ as follows. 
For each recursive instance that contains $e$, we add the vertices in the central bag to $V(e)$. 
Since $e$ is contained in $O(\log n)$ instances (from the construction in Section~\ref{sec:quasipartition-tw}), we have $|V(e)|=O(\tw \log n)$.
Note that $V(e)$ is independent of the scale $i$.

Now, we reinterpret the construction in Section~\ref{sec:quasipartition-tw} as follows. For every scale $i$, we sample a radius $R_i$ uniformly at random from $[2^{i-1},2^i]$, and we add $e$ to the cutset $S_i$ 
if $e$ is contained in one of the boundaries of the (in- and out-) balls of radius $R_i\in [2^{i-1}, 2^{i}]$ centered at some vertex in $V(e)$. This defines the quasipartition $\mathcal D_i$ over the cutsets.
We prove the following useful lemma.
\begin{lemma} \label{lem:quasipartition-correct1}
    Let $e=(x,y)$ be an edge and $z$ be a vertex in $V(e)$. Then,
    there are at most $O(1)$ scales $i\in [k]$ satisfying the following conditions:
    \begin{itemize}
        \item[$(a)$] $\Pr[e\in \delta^+(B^+_G(z, R_i))] > 0$, and
        \item[$(b)$] $2^{i} > w(e)$.
    \end{itemize}
\end{lemma}
\begin{proof}
    Let $i$ be an index satisfying conditions $(a),(b)$.
    Note that $e$ is contained in  $\delta^+(B^+_G(z, R_i))$ if and only if $d_G(z,x) \leq R_i$ and $d_G(z, y)>R_i$. Subsequently, from condition $(a)$, we have that:
    $$\Pr_{R_i}[(d_G(z,x) \leq R_i) \cap (d_G(z, y) > R_i))] >0.$$
    Since $R_i$ is sampled uniformly at random from $[2^{i-1}, 2^i]$, we have that:
    $$d_G(z,x) <2^i,\text{ and }   d_G(z,y) \geq 2^{i-1}.$$
    By the triangle inequality and condition $(b)$, we obtain:
    $$2^{i-1} \leq d_G(z,y) \leq d_G(z,x) + d_G(x,y) \leq d_G(z,x) + w(e) < 2^i+2^i = 2^{i+1}.$$
    Let $j$ be the fixed integer such that $d_G(z,y)\in [2^j, 2^{j+1}]$.
    Then, $i\in \{j-1,j,j+1,j+2\}$.
    We conclude that there are at most 4 scales $i$ satisfying both conditions $(a),(b)$. This completes the proof.
\end{proof}
Using exactly the same argument, we can show the following lemma.
\begin{lemma} \label{lem:quasipartition-correct2}
    Let $e=(x,y)$ be an edge and $z$ be a vertex in $V(e)$. Then,
    there are at most $O(1)$ scales $i\in [k]$ satisfying the following conditions:
    \begin{itemize}
        \item[$(a)$] $\Pr[e\in \delta^-(B^-_G(z, R_i))] > 0$, and
        \item[$(b)$] $2^{i} > w(e)$.
    \end{itemize}
\end{lemma}

For an edge $e$ and a vertex $z\in V(e)$, let $I(z,e)$ be the set of indices $i$ such that $\Pr[e\in \delta^+(B^+_G(z, R_i)) \cup \delta^-(B^-_G(z, R_i))]>0$ and $2^i > w(e)$. Due to~\Cref{lem:quasipartition-correct1,lem:quasipartition-correct2}, we have that:
\begin{align} \label{eq:18-13}
    |I(z,e)|=O(1).
\end{align}

Let $i\in I(z,e)$. Since $R_i$ is chosen uniformly at random from $[2^{i-1}, 2^i]$, we have that:
\begin{align} \label{eq:18-14}
    \Pr_{R_i}[(d_G(z,x) \leq R_i) \cap (d_G(z, y) > R_i))] \leq O(1)\cdot \frac{w(e)}{2^i}. 
\end{align}

Next, we analyze the probability $q:=\Pr_{S\sim \mathcal D}[\text{there is no } u\text{-}v \text{ path in } G\setminus S]$. Our argument closely follows that of~\cite{memoli2018quasimetric}.
Let $\mathcal P=\{p_1,p_2,\ldots, p_\ell\}$ be the set of all $u$-$v$ paths in $G$.
For every $m\in [\ell]$, let $X_m$ be the event that $p_m$ contains at least one edge in $S$.
Then, $q= \Pr[X_1\wedge X_2\wedge \ldots \wedge X_\ell] \leq \Pr[X_m]$ for all $m\in [\ell]$.
Now let us consider a fixed \emph{shortest path} $p_m=\langle u=v_1,v_2,\ldots, v_t=v\rangle\in \mathcal P$.
For every $j\in [t-1]$, let $Y_j$ be the event that $(v_j, v_{j+1})\in S$.
Then, $\Pr[X_m] =\Pr[Y_1\vee Y_2\vee\ldots \vee Y_{t-1}]\leq \sum_{j}\Pr[Y_j]$ due to the union bound of probability. In addition, let $e_j=(v_j, v_{j+1})$. Then,
\begin{align}
    \Pr[Y_j]=\Pr[e_j\in S] &= \sum_{i=0}^k \frac{2^i}{c} \cdot \Pr[e_j\in S_i]
    \\&\leq \sum_{i=0}^k \frac{2^i}{c} \cdot \sum_{z\in V(e_j)}  \Pr_{R_i}[(d_G(z,v_j) \leq R_i) \cap (d_G(z, v_{j+1}) > R_i))]
    \\&= \sum_{z\in V(e_j)} \sum_{i=0}^k \frac{2^i}{c} \cdot   \Pr_{R_i}[(d_G(z,v_j) \leq R_i) \cap (d_G(z, v_{j+1}) > R_i))]. \label{eq:18-12}
\end{align}
We analyze the inner summation over $i$ as follows. We split the sum over $k$ scales at $i=\lfloor \log w(e_j)\rfloor$. Since $I(z,e_j)\subseteq \{i:2^i>w(e_j)\}$, we can bound each part separately. Let $P_{i,z}:=\Pr_{R_i}[(d_G(z,v_j) \leq R_i) \cap (d_G(z, v_{j+1}) > R_i))]$. Then,
\begin{align*}
    &\sum_{i=0}^k \frac{2^i}{c} \cdot \Pr_{R_i}[(d_G(z,v_j) \leq R_i) \cap (d_G(z, v_{j+1}) > R_i))] = \sum_{i=0}^k \frac{2^i}{c} \cdot P_{i,z}
    \\&=\sum_{i=0}^{\lfloor\log w(e_j) \rfloor} \frac{2^i}{c} \cdot P_{i,z} + \sum_{i=\lfloor \log w(e_j)\rfloor +1}^k \frac{2^i}{c} \cdot  P_{i,z}
    \\&\leq \sum_{i=0}^{\lfloor \log w(e_j) \rfloor} \frac{2^i}{c}+  \sum_{i\in I(z, e_j)} \frac{2^i}{c}\cdot P_{i,z}
     \\&\stackrel{\eqref{eq:18-14}}\leq \frac{2\cdot w(e_j)}{c} + \sum_{i\in I(z, e_j)} \frac{2^i}{c} \cdot (O(1)\cdot \frac{w(e_j)}{2^i}) \stackrel{\eqref{eq:18-13}}= O(1)\cdot \frac{w(e_j)}{c}.
\end{align*}
Returning to~\Cref{eq:18-12}, we obtain:
\begin{align*}
    \Pr[Y_j] &\leq \sum_{z\in V(e_j)} \sum_{i=0}^k \frac{2^i}{c} \cdot   \Pr_{R_i}[(d_G(z,v_j) \leq R_i) \cap (d_G(z, v_{j+1}) > R_i))].
    \\&\leq \sum_{z\in V(e_j)} O(1)\cdot \frac{w(e_j)}{c}=O(|V(e_j)|)\cdot \frac{w(e_j)}{c} = O(\tw \log n) \cdot \frac{w(e_j)}{c}.
\end{align*}
Here, the last equality follows from the fact that $|V(e)|=O(\tw \log n)$ for every edge $e$.
Since $p_m$ is a shortest $u$-$v$ path in $G$, we have that:
\begin{align} \label{eq:18-15}
    q \leq \Pr[X_m] \leq \sum_j \Pr[Y_j] \leq O(\tw \log n)\cdot \sum_j \frac{w(e_j)}{c} = O(\tw\log n)\cdot \frac{d_G(u, v)}{c}.
\end{align}
Putting everything together, we show that the probability bound from~Eq.\eqref{eq:5-16} to Eq.\eqref{eq:5-20} can be improved as follows:
\begin{align*}
    d^*(u,v)
    = \Pr_{S\sim \mathcal D}[\text{there is no } u\text{-}v \text{ path in } G\setminus S] 
    \stackrel{\eqref{eq:18-15}}\leq O(\tw \log n) \cdot \frac{d_G(u,v)}{c}.
\end{align*}
This confirms the improved distortion guarantee for the algorithm of~\cite{memoli2018quasimetric}.
\paragraph{Remark.}
In this section, we only analyzed the case where every edge weight is an integer and bounded by $\text{poly}(n)$.
However, our idea can be lifted to the general case as well: the problematic distortion loss comes from the scales that are much larger than $w(e)$. In~\cite{memoli2018quasimetric}, this issue was addressed by introducing a so-called \emph{forced quasipartition}. Roughly speaking, in a forced quasipartition, it is guaranteed that $\Pr[e\in S]$ for a cutset $S$ at scale $i$ becomes zero when $\frac{2^i}{w(e)} \geq n$ (assuming that $n$ is power of two).
In this way, they bounded their distortion by summing up $O(\log n)$ scales $i$ from $\lfloor \log w(e) \rfloor $ to $\lfloor \log (n\cdot w(e)) \rfloor \leq  \lfloor \log w(e)\rfloor+\log n$.

In our argument, rather than considering the fixed scales between $\lfloor \log w(e)\rfloor$ and $\lfloor \log w(e) \rfloor+\log n$, we consider the probability $\Pr_{R_i}[e\in \delta^+(B^+_G(z, R_i))\cup \delta^-(B^-_G(z, R_i))]$ for every $z\in V(e)$, which is nonzero for only $|I(z,e)|=O(1)$ scales.
Our $O(\tw \log n)$ distortion factor is derived solely from the fact that $|V(e)| = O(\tw \log n)$ for every edge $e$.
In this way, we can remove the bounded weights assumption and conclude the following lemma.
\begin{lemma}
    Given a directed graph $G$ having treewidth $\tw$, there is an embedding of $G$ into the convex combination of 0-1 quasimetric space with distortion $O(\tw \log n)$.
\end{lemma}
By applying this improved bound to the reduction to the integrality gap of the sparsest cut problem~\cite{charikar2006directed}, we obtain the main result of this section.
\flowcutgap*

\bibliographystyle{alpha}
\bibliography{paper,DAGbib}

\end{document}